\documentclass[runningheads]{llncs}
\usepackage{amsmath}
\usepackage{amssymb}

\usepackage{cite}

\DeclareRobustCommand{\VAN}[3]{#2} 

\usepackage{color}
\usepackage{tikz}
\usetikzlibrary{positioning}
\usetikzlibrary{shapes}
\usetikzlibrary{patterns}
\usetikzlibrary{arrows}
\usetikzlibrary{decorations.pathmorphing}
\tikzset{every picture/.style={>=stealth'}}

\usepackage{graphicx}
\usepackage{pgfplots}

\newcommand{\tikzdot}[1]{
  \protect\tikz[baseline=-3.4pt]{
    \fill[color=#1, opacity=0.7] circle[radius=1.7pt];
  }
}

\makeatletter
\def\addlegendimage{\pgfplots@addlegendimage}
\makeatother

\tikzstyle{dots_significant}=[only marks, mark=o, mark size=1.8pt, mark options={color=black, opacity=0.9}]

\colorlet{qbf}{cyan}
\tikzstyle{dots_qbf}=[only marks, mark=*, mark size=1.8pt, mark options={color=qbf, fill=qbf, opacity=0.5}]
\tikzstyle{x_qbf}=[only marks, mark=x, mark size=3pt, mark options={color=qbf!80!black, fill=qbf, opacity=0.9}]
\tikzstyle{plot_qbf}=[color=qbf, dashed, line width=0.7pt, opacity=0.7]

\colorlet{goe}{orange}
\tikzstyle{dots_goe}=[only marks, mark=*, mark size=1.8pt, mark options={color=goe, fill=goe, opacity=0.5}]
\tikzstyle{x_goe}=[only marks, mark=x, mark size=3pt, mark options={color=goe!80!black, fill=goe, opacity=0.9}]
\tikzstyle{plot_goe}=[color=goe, dashed, line width=0.7pt, opacity=0.7]

\usepackage{listings}
\lstdefinelanguage{pseudo}{
  morekeywords=[1]{
    break, continue, each, else, for, if, loop, let, otherwise, repeat, return,
    then, until, while, NIL, Apply, Reduce, },
  morecomment=[l]{//},
  numbers=left,
  literate={:=}{{$\gets$}}1
  {<=}{{$\leq$}}1
  {>=}{{$\geq$}}1
  {<>}{{$\neq$}}1
  {!}{{$\neg$}}1
  {->}{{$\rightarrow$}}1
  {=>}{{$\Rightarrow$}}1
}

\usepackage{hyperref,doi}
\usepackage[capitalise]{cleveref}
\usepackage{caption,subcaption}
\usepackage{multirow}
\hypersetup{
    colorlinks,
    linkcolor={red!50!black},
    citecolor={green!40!black},
    urlcolor={blue!40!black}
}

\usepackage{cite}

\usepackage{graphicx}
\newcommand{\B}[0]{\ensuremath{\mathbb{B}}}

\newcommand{\Oh}[1]{\ensuremath{\mathcal{O} ( #1 )}}

\newcommand{\scan}[1]{\text{scan} ( #1 )}
\newcommand{\sort}[1]{\text{sort} ( #1 )}

\newcommand{\PQ}[2]{\ensuremath{Q_{\mathit{#1}:\mathit{#2}}}}
\newcommand{\LIST}[2]{\ensuremath{L_{\mathit{#1}:\mathit{#2}}}}
\newcommand{\FILE}[1]{\ensuremath{F_{\mathit{#1}}}}

\newcommand{\arc}[3][solid]{
  \ensuremath{#2}
  \,
  \tikz[baseline=-\the\dimexpr\fontdimen22\textfont2\relax]{
    \draw[->, #1](0,0) -- ++(1.2em,0);
  }
  \,
  \ensuremath{#3}
}

\newcommand{\uidof}[1]{\ensuremath{#1}.\texttt{uid}}
\newcommand{\topof}[1]{\ensuremath{#1}.\texttt{var}}
\newcommand{\lowof}[1]{\ensuremath{#1.\texttt{low}}}
\newcommand{\highof}[1]{\ensuremath{#1.\texttt{high}}}

\def\orcidID#1{\smash{\href{http://orcid.org/#1}{\protect\raisebox{-1.25pt}{\protect\includegraphics{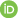}}}}}

\newcommand*{\mailto}[1]{\href{mailto:#1}{\nolinkurl{#1}}}

\usepackage{bbding}

\def\arxiv{1}

\if\arxiv1%
  \newcommand{\Nprime}[0]{\ensuremath{N'}}
\else%
  \newcommand{\Nprime}[0]{\ensuremath{N}}
\fi%

\title{Multi-variable Quantification of BDDs\\in External Memory using Nested Sweeping}
\if\arxiv1%
  \subtitle{(Extended Version)}
\fi%

\titlerunning{Multi-variable Quantification of BDDs in External Memory}

\author
{
  Steffan Christ S{\o}lvsten
  \if\arxiv0
    (\Envelope)
  \fi
  \orcidID{0000-0003-0963-6569}
  \and
  Jaco van de Pol
  \orcidID{0000-0003-4305-0625}
}

\authorrunning
{
  S. C. S{\o}lvsten and J. van de Pol
}

\institute
{
  Aarhus University, Denmark
  \texttt{\{%
    \href{mailto:soelvsten@cs.au.dk}{\color{black} soelvsten},%
    \href{mailto:jaco@cs.au.dk}{\color{black} jaco}%
    \}@cs.au.dk}
}

\begin{document}

\maketitle


\begin{abstract}
  Previous research on the Adiar BDD package has been successful at designing algorithms capable of
  handling large Binary Decision Diagrams (BDDs) stored in external memory. To do so, it uses
  consecutive sweeps through the BDDs to resolve computations. Yet, this approach has kept
  algorithms for multi-variable quantification, the relational product, and variable reordering out
  of its scope.

  In this work, we address this by introducing the \emph{nested sweeping} framework. Here, multiple
  concurrent sweeps pass information between each other to compute the result. We have implemented
  the framework in Adiar and used it to create a new external memory multi-variable quantification
  algorithm. In practice, this improves Adiar's running time by a factor of $1.7$. In turn, this
  work extends the previous research results on Adiar to also apply to its quantification operation:
  compared to conventional depth-first implementations, Adiar with nested sweeping is able to solve
  more problems and/or solve them faster.

  \if\arxiv0
    \keywords{
      Time-forward Processing \and
      External Memory Algorithms \and
      Binary Decision Diagrams
    }
  \fi
\end{abstract}


\section{Introduction} \label{sec:introduction}

The ability of Binary Decision Diagrams (BDDs) to represent Boolean formulae as small directed
acyclic graphs (DAGs) have made them an invaluable tool to solve many complex problems. For example,
recently they have been used to check type-and-effect systems \cite{Madsen2020,Madsen2023}, to
generate proofs for SAT and QBF solvers \cite{Bryant2021:SAT,Bryant2021:QBF,Bryant2022}, for circuit
synthesis \cite{Fried2016,Yi2022}, to solve games~\cite{Dijk2024,Michaud2018,Renkin2022}, and for
symbolic model checking
\cite{Cimatti2000,Gammie2004,Ciardo2009,Kant2015,Lomuscio2017,He2020,Amparore2022}.

Implementations of decision diagrams conventionally make use of recursive depth-first algorithms and
a unique node table \cite{Brace1990,Minato1990,Karplus1988,Somenzi2015,Lind1999,Dijk2016}. Both of
these introduce random access, which pauses the entire computation while missing data is fetched
\cite{Klarlund1996,Minato2001,Pastva2023}. For large enough instances, data has to reside on disk
and the resulting I/O-operations that ensue become the bottle-neck.

Adiar~\cite{Soelvsten2022:TACAS} is a BDD package written in C++ based on the ideas of Lars
Arge~\cite{Arge1995:1}: the depth-first recursive algorithms are replaced with iterative algorithms.
Here, one or more priority queues reorder the execution of recursive calls such that they are
synchronised with a level-by-level traversal of the inputs. This makes Adiar's algorithms, unlike
the conventional recursive implementations, optimal in the I/O-model~\cite{Aggarwal1987} of Aggarwal
and Vitter \cite{Arge1995:1,Arge1996}. In turn, this enables it to manipulate BDDs beyond the reach
of conventional BDD packages at a negligible cost to its running time \cite{Soelvsten2022:TACAS}.

Yet, the ideas in \cite{Arge1995:1,Arge1996,Soelvsten2022:TACAS} only provide a translation of the
simplest BDD algorithms, which do not recurse on the result of other recursive calls.
This does not provide a
way to translate the more complex BDD algorithms that recurse on intermediate recursion results,
e.g.\ multi-variable quantification. Hence, until this work, Adiar could not easily be used for
solving Quantified Boolean formul{\ae} (QBF). Furthermore, game solving and symbolic model checking
has until now been out of reach for Adiar.

\subsection{Contributions}

In \cref{sec:theory}, we introduce the notion of \emph{nested sweeping} to provide a framework on
which these more complex BDD operations can be implemented. Here, an \emph{outer} bottom-up sweep
accumulates the results from multiple nested \emph{inner} sweeps. With this framework in hand, we
implement an I/O-efficient multi-variable quantification akin to the one in conventional BDD
packages. %
\if\arxiv1%
  Furthermore, we identify in \cref{sec:theory:nested optimisations} optimisations for the nested
  sweeping framework in general and in \cref{sec:theory:quantify optimisations} for the
  quantification operation in particular. \cref{sec:implementation} provides an overview of the
  implementation while \cref{sec:experiments} %
\else%
  Furthermore, we identify in \cref{sec:theory:nested optimisations} optimisations for the nested
  sweeping framework. The full paper~\cite{Soelvsten2024:arXiv} also includes optimisations specific
  to the quantification operation and an overview of the implementation in Adiar.
  \Cref{sec:experiments} %
\fi%
shows that nested sweeping improves the running time in practice by a factor of $1.7$ when solving
QBF-encodings of two-player games and when reasoning about the transition system in Conway's Game of
Life~\cite{Martin1970}. We compare our approach to related work in \cref{sec:related work} and
provide our conclusions and future work in \cref{sec:conclusion}.

\section{Preliminaries} \label{sec:preliminaries}

\subsection{The I/O-Model} \label{sec:preliminaries:io}

Aggarwal and Vitter introduced the I/O-model~\cite{Aggarwal1987} to analyse the cost of transferring
data to and from a slow storage device. Here, computations can only operate on data that resides in
\emph{internal} memory, e.g. the RAM, with a finite size of $M$. Hence, if the input of size $N$ (or
some intermediate result) exceeds $M$ then it needs to be transferred to and from \emph{external}
memory, e.g. the disk. Yet, each such data transfer (I/O) moves an entire consecutive block of $B$
elements; an algorithm's I/O-complexity is the number of I/Os it uses.

One needs $\scan{N} \triangleq N/B$ I/Os to linearly scan through a consecutive list of $N$ elements
in external memory \cite{Aggarwal1987}. Assuming $N > M$, one needs to use $\Theta(\sort{N})$ I/Os
to sort $N$ elements, where $\sort{N} \triangleq N/B \cdot \log_{M/B}(N/B)$ \cite{Aggarwal1987}.
Furthermore, one can design an I/O-efficient priority queue capable of doing $N$ insertions and
deletions in $\Theta(\sort{N})$ I/Os \cite{Arge1995:2}.
For simplicity, we overload $\scan{N}$ to be $N$ and $\sort{N}$ to be $N \log_2 N$ when referring to
an algorithm's time complexity rather than its I/O complexity.

Intuitively, an algorithm is I/O-inefficient if it uses an entire I/O to retrieve a block but does
not make use of a significant portion of the $B$ elements within.
That is, random access can result in $N$ I/Os. For all realistic values of $N$, $M$, and
$B$, this is several magnitudes larger than both $\scan{N}$ and $\sort{N}$.

\subsection{Binary Decision Diagrams} \label{sec:preliminaries:bdd}

\if\arxiv1
  \begin{figure}[t]
\else
  \begin{figure}[t]
\fi
  \centering

  \subfloat[$x_0$]{
    \label{fig:bdd_example:x0}
    \centering

    \begin{tikzpicture}
      \node[shape = circle,    draw = black] at (0, 1.4) (r) {$x_0$};

      \node[shape = rectangle, draw = black] at (-0.8, 0) (F) {$\bot$};
      \node[shape = rectangle, draw = black] at ( 0.8, 0) (T) {$\top$};

      \draw[->, dashed]
        (r) edge (F)
      ;
      \draw[->]
        (r) edge (T)
      ;
    \end{tikzpicture}
  }
  \quad
  \subfloat[$x_1$]{
    \label{fig:bdd_example:x1}
    \centering

    \begin{tikzpicture}
      \node[shape = circle,    draw = black] at (0, 0.8) (r) {$x_1$};

      \node[shape = rectangle, draw = black] at (-0.8, 0) (F) {$\bot$};
      \node[shape = rectangle, draw = black] at ( 0.8, 0) (T) {$\top$};

      \draw[->, dashed]
        (r) edge (F)
      ;
      \draw[->]
        (r) edge (T)
      ;
    \end{tikzpicture}
  }
  \quad
  \subfloat[$\neg x_1$]{
    \label{fig:bdd_example:-x1}
    \centering

    \begin{tikzpicture}
      \node[shape = circle,    draw = black] at (0, 0.8) (r) {$x_1$};

      \node[shape = rectangle, draw = black] at (-0.8, 0) (F) {$\bot$};
      \node[shape = rectangle, draw = black] at ( 0.8, 0) (T) {$\top$};

      \draw[->, dashed]
        (r) edge (T)
      ;
      \draw[->]
        (r) edge (F)
      ;
    \end{tikzpicture}
  }
  \quad
  \subfloat[$x_0 \land \neg x_1$]{
    \label{fig:bdd_example:x0&-x1}
    \centering

    \begin{tikzpicture}
      \node[shape = circle,    draw = black] at (-0.4, 1.4) (r0) {$x_0$};
      \node[shape = circle,    draw = black] at ( 0.4, 0.8) (r1) {$x_1$};

      \node[shape = rectangle, draw = black] at (-0.8, 0) (F) {$\bot$};
      \node[shape = rectangle, draw = black] at ( 0.8, 0) (T) {$\top$};

      \draw[->, dashed]
        (r0) edge (F)
        (r1) edge (T)
      ;
      \draw[->]
        (r0) edge (r1)
        (r1) edge (F)
      ;
    \end{tikzpicture}
  }

  \caption{Examples of Reduced Ordered Binary Decision Diagrams. Terminals are drawn as boxes
    surrounding their Boolean value. Internal nodes are drawn as circles and contain their decision
    variable. Arcs to the \emph{high} and \emph{low} child are respectively drawn solid and dashed.}
  \label{fig:bdd_example}
\end{figure}

As shown in \cref{fig:bdd_example}, a Binary Decision Diagram~\cite{Bryant1986} (BDD) (based on
\cite{Lee1959,Akers1978}) represents an $n$-ary Boolean function as a singly-rooted directed acyclic
graph (DAG). Each of its two sinks, refered to as \emph{terminals}, contain one of the two Boolean
values, $\B = \{ \top, \bot \}$. These represent the function's output values. An internal BDD node,
$v$, is associated in $\topof{v}$ with a Boolean input variable $x_i$. Furthermore, it has two BDD
nodes as children, $\lowof{v}$ and $\highof{v}$. These three values in $f$ encode the ternary
if-then-else
\if\arxiv1%
  \begin{equation*}
    \topof{v} \ ?\ \highof{v} : \lowof{v} \enspace .
  \end{equation*}

\else%
  $\topof{v} \ ?\ \highof{v} : \lowof{v}$. %
\fi%
What are colloquially referred to as BDDs are in fact \emph{Reduced Ordered} Binary Decision
Diagrams (ROBDDs). An Ordered BDD (OBDD) restricts each variable to occur at most once on each path
from the root to a terminal and to occur according to a certain order, $\pi$. This gives rise to a
levelisation of the OBDD where each level, $\ell$, is associated with an input variable, $x_i$. For
sake of simplicity, we assume that $\pi$ is the identity order. A \emph{Reduced} OBDD further
restricts the DAG such that (1) no nodes are duplicates of another and (2) no node is redundant,
i.e.\ $\highof{v} = \lowof{v}$. Assuming the variable ordering, $\pi$, is fixed, ROBDDs are a unique
canonical form of the Boolean function it represents.

\subsubsection*{Quantification Algorithm}

The levelisation of OBDDs allows the recursive BDD algorithms to both be efficient and elegant. For
example, the \texttt{or} operation works by a product construction of the two input BDDs. Here, each
node of the output BDD simulates, according to $\pi$, the decision(s) taken on the shallowest BDD
node(s) in the product of nodes from the input.

\if\arxiv1
  \begin{figure}[t]
\else
  \begin{figure}[t]
\fi
  \centering

  \begin{lstlisting}
*@{\bf exists}@*($v$, $X$)
  if $v = \bot \lor v = \top$
      return $v$
  exi0 := *@{\bf exists}@*($\lowof{v}$, $X$)
  exi1 := *@{\bf exists}@*($\highof{v}$, $X$)
  if $\topof{v} \not\in X$
      return Node { $\topof{v}$, exi0, exi1 }
  return *@{\bf or}@*(exi0, exi1)
  \end{lstlisting}

  \caption{A recursive multi-variable {\bf exists} operation.}
  \label{fig:exists}
\end{figure}

Since $(\exists x : \phi) \equiv \phi[\top / x] \lor \phi[\bot / x]$, the \texttt{or} operation can
be used as the basis for an existential quantification ($\exists$) for a set of input variables,
$X = \{ x_i, x_j, \dots, x_k \}$. As shown in \cref{fig:exists}, if $v$ is a terminal then this
(sub)BDD depends on none of the to be quantified variables. Otherwise, both its children are
resolved recursively into intermediate results, \texttt{exi0} and \texttt{exi1}. If the decision
variable of the root, $\topof{v}$, should not be quantified, a new node with variable $\topof{v}$ is
created from the two recursive results. Otherwise, \texttt{exi0} and \texttt{exi1} are instead
combined (recursively once more) with a nested \texttt{or} operation.

Similarly, one can implement a universal quantification ($\forall$) by use of a nested \texttt{and}
operation. For clarity, our contributions in \cref{sec:theory} are only phrased with respect to the
\texttt{exists} operation. But, everything that follows also applies to \texttt{forall} by replacing
\texttt{or} with \texttt{and}.

\if\arxiv1%
  \subsubsection*{Relational Product}

  The \emph{relational product} computes the set of states after taking a step in a transition
  system with the formula $\exists \vec{x} : S(\vec{x}) \wedge R(\vec{x},\vec{x'})$. Hence, the
  support for a multi-variable quantification operation is key for the application of BDDs in the
  context of symbolic model checking.

\fi%

\subsection{I/O-efficient BDD Manipulation} \label{sec:preliminaries:adiar}

\if\arxiv1
  \begin{figure}[b]
\else
  \begin{figure}[b]
\fi
  \centering

  \subfloat[Node-based representation of $x_0 \land \neg x_1$ (\cref{fig:bdd_example:x0&-x1}).]{
    \label{fig:bdd_example:representation:node}
    \centering

    \qquad\qquad
    $\left[
      \
      \{(0,0), \bot, (1,0)\};
      \
      \{(1,0), \top, \bot\}
      \
    \right]$
    \qquad\qquad
  }

  \vspace{10pt}

  \subfloat[Arc-based representation of the \texttt{or} of $x_0 \land \neg x_1$
  (\cref{fig:bdd_example:x0&-x1}) and $x_1$ (\cref{fig:bdd_example:x1}).]{
    \label{fig:bdd_example:representation:arc}
    \centering

    $\left[
      \
      \arc[solid]{(0,0)}{(1,0)};
      \
      \arc[dashed]{(0,0)}{(1,1)};
      \
      \arc[dashed]{(1,0)}{\bot};
      \
      \arc[solid]{(1,0)}{\top};
      \
      \arc[dashed]{(1,1)}{\top};
      \
      \arc[solid]{(1,1)}{\top};
      \
    \right]$
  }

  \caption{BDD Representations in Adiar.}
  \label{fig:bdd_example:representation}
\end{figure}

The Adiar~\cite{Soelvsten2022:TACAS} BDD package builds on top of Lars Arge's ideas
\cite{Arge1995:1,Arge1996} on how to improve the I/O complexity of BDD manipulation. To not
introduce random access, Adiar does not use any hash tables nor recursion for its BDD manipulation.
As a result, different BDD objects do not share common subtrees in Adiar. For the same reason, it
neither uses pointers to traverse its BDDs. Instead, every BDD node $v$ is uniquely identified by a
pair $(\topof{v}, v.\texttt{id})$ where $v.\texttt{id}$ is $v$'s index on level $\topof{v}$.
Lexicographically, this \emph{unique identifier} (\texttt{uid}) imposes a total ordering of all BDD
nodes such that they follow the variable ordering. Yet, the \texttt{uid} does not specify the exact
index where one can find the BDD node.
For example, the BDD for $x_0 \land \neg x_1$ in \cref{fig:bdd_example:x0&-x1} is represented in
Adiar as the list of nodes in \cref{fig:bdd_example:representation:node}: every node is a 3-tuple
with its \texttt{uid} followed by the unique identifier of its low and its high children.

\if\arxiv1
  \begin{figure}[t]
\else
  \begin{figure}[t]
\fi
  \centering

  \begin{tikzpicture}[every text node part/.style={align=center}]
    \draw (0,0) rectangle ++(2,1)
    node[pos=.5]{\texttt{Apply}};
    \draw (4.5,0) rectangle ++(2,1)
    node[pos=.5]{\texttt{Reduce}};

    \draw[->] (-0.5,0.8) -- ++(0.5,0)
    node[pos=-1.3]{$f$ \texttt{nodes}};
    \draw[->] (-0.5,0.2) -- ++(0.5,0)
    node[pos=-1.3]{$g$ \texttt{nodes}};

    \draw[->,dashed] (2,0.5) -- ++(2.5,0)
    node[pos=0.5,above]{\small $f \lor g$ \texttt{arcs}};

    \draw[->] (6.5,0.5) -- ++(0.5,0)
    node[pos=2.9]{$f \lor g$ \texttt{nodes}};
  \end{tikzpicture}

  \caption{The Apply--Reduce pipeline of \texttt{or} in Adiar.}
  \label{fig:tandem}
\end{figure}
\if\arxiv1
  \begin{figure}[!t]
\else
  \begin{figure}[t]
\fi
  \centering

  \subfloat[Apply~($x_0$)]{
    \label{fig:or:example:apply:x0}
    \centering

    \begin{tikzpicture}
      \draw[red, dashed] (-0.75, 1.4) -- ++(3.1,0);

      \node[shape = circle,    draw = lightgray, lightgray, fill=white] at ( 0.8, 1.4) (n1) {$x_0$};

      \node[] at ( 0.0, 0.8) (n2) {};
      \node[] at ( 1.6, 0.8) (n3) {};

      \node[shape = rectangle] at ( 0.0, 0) (F) {\phantom{$\bot$}};
      \node[shape = rectangle] at ( 1.6, 0) (T) {\phantom{$\top$}};

      \draw[<-, dashed, lightgray]
        (n1) edge (n2)
      ;
      \draw[<-, lightgray]
        (n1) edge (n3)
      ;
    \end{tikzpicture}
  }
  \subfloat[Apply~($x_1$)]{
    \label{fig:or:example:apply:x1}
    \centering

    \begin{tikzpicture}
      \draw[red, dashed] (-0.75, 0.8) -- ++(3.1,0);

      \node[shape = circle,    draw = black] at ( 0.8, 1.4) (n1) {$x_0$};

      \node[shape = circle,    draw = lightgray, lightgray, fill=white] at ( 0.0, 0.8) (n2) {$x_1$};
      \node[shape = circle,    draw = lightgray, lightgray, fill=white] at ( 1.6, 0.8) (n3) {$x_1$};

      \node[shape = rectangle] at ( 0.0, 0) (F) {\phantom{$\bot$}};
      \node[shape = rectangle] at ( 1.6, 0) (T) {\phantom{$\top$}};

      \draw[<-, dashed]
        (n1) edge (n2)
      ;
      \draw[<-, dashed, lightgray]
        (n2) edge (F)
        (n3) edge[bend right] (T)
      ;
      \draw[<-]
        (n1) edge (n3)
      ;
      \draw[<-, lightgray]
        (n2) edge (T)
        (n3) edge[bend left] (T)
      ;
    \end{tikzpicture}
  }
  \subfloat[Apply (leaves)]{
    \label{fig:or:example:apply:terms}
    \centering

    \begin{tikzpicture}
      \draw[red, dashed] (-0.75, 0.0) -- ++(3.1,0);

      \node[shape = circle,    draw = black] at ( 0.8, 1.4) (n1) {$x_0$};

      \node[shape = circle,    draw = black] at ( 0.0, 0.8) (n2) {$x_1$};
      \node[shape = circle,    draw = black] at ( 1.6, 0.8) (n3) {$x_1$};

      \node[shape = rectangle, draw = black, fill=white] at ( 0.0, 0) (F) {$\bot$};
      \node[shape = rectangle, draw = black, fill=white] at ( 1.6, 0) (T) {$\top$};

      \draw[<-, dashed]
        (n1) edge (n2)
        (n2) edge (F)
        (n3) edge[bend right] (T)
      ;
      \draw[<-]
        (n1) edge (n3)
        (n2) edge (T)
        (n3) edge[bend left] (T)
      ;
    \end{tikzpicture}
  }
  \\ \bigskip
  \subfloat[Reduce~(leaves)]{
    \label{fig:or:example:reduce:terms}
    \centering

    \begin{tikzpicture}
      \draw[red, dashed] (-0.75, 0.0) -- ++(3.1,0);

      \node[shape = circle,    draw = lightgray, lightgray] at ( 0.0, 0.8) (n2) {$x_1$};
      \node[shape = circle,    draw = lightgray, lightgray] at ( 1.6, 0.8) (n3) {$x_1$};

      \node[shape = rectangle, draw = black, fill=white] at ( 0.0, 0) (F) {$\bot$};
      \node[shape = rectangle, draw = black, fill=white] at ( 1.6, 0) (T) {$\top$};

      \draw[->, dashed, lightgray]
        (n2) edge (F)
        (n3) edge[bend right] (T)
      ;
      \draw[->, lightgray]
        (n2) edge (T)
        (n3) edge[bend left] (T)
      ;
    \end{tikzpicture}
  }
  \subfloat[Reduce~($x_1$)]{
    \label{fig:or:example:reduce:x1}
    \centering

    \begin{tikzpicture}
      \draw[red, dashed] (-0.75, 0.8) -- ++(3.1,0);

      \node[shape = circle,    draw = lightgray, lightgray] at ( 0.8, 1.4) (n1) {$x_0$};
      \node[shape = circle,    draw = black, fill=white] at ( 0.0, 0.8) (n2) {$x_1$};

      \node[shape = rectangle, draw = black] at ( 0.0, 0) (F) {$\bot$};
      \node[shape = rectangle, draw = black] at ( 1.6, 0) (T) {$\top$};

      \draw[->, dashed, lightgray]
        (n1) edge (n2)
      ;
      \draw[->, dashed]
        (n2) edge (F)
      ;
      \draw[->, lightgray]
        (n1) edge (T)
      ;
      \draw[->]
        (n2) edge (T)
      ;
    \end{tikzpicture}
  }
  \subfloat[Reduce~($x_0$)]{
    \label{fig:or:example:reduce:x0}
    \centering

    \begin{tikzpicture}
      \draw[red, dashed] (-0.75, 1.4) -- ++(3.1,0);

      \node[shape = circle,    draw = black, fill=white] at ( 0.8, 1.4) (n1) {$x_0$};
      \node[shape = circle,    draw = black] at ( 0.0, 0.8) (n2) {$x_1$};

      \node[shape = rectangle, draw = black] at ( 0.0, 0) (F) {$\bot$};
      \node[shape = rectangle, draw = black] at ( 1.6, 0) (T) {$\top$};

      \draw[->, dashed]
        (n1) edge (n2)
        (n2) edge (F)
      ;
      \draw[->]
        (n1) edge (T)
        (n2) edge (T)
      ;
    \end{tikzpicture}
  }

  \caption{Step-by-step example of the \texttt{or} of $x_1$ (\cref{fig:bdd_example:x1}) and
    $x_0 \land \neg x_1$ (\cref{fig:bdd_example:x0&-x1}) with time-forward processing; subfigures
    show the state after processing each level. Arcs in gray are pushed to the algorithm's priority
    queue whereas the ones in black have been written to the output file.}
  \label{fig:or:example}
\end{figure}

As depicted in \cref{fig:tandem,fig:or:example}, the previous BDD operations in Adiar, such as
\texttt{or}, process a BDD with two sweeps. Both sweeps use \emph{time-forward
  processing}~\cite{Chiang1995,Arge1995:2} to achieve their I/O-efficiency: computation is deferred
with one or more priority queues until all relevant data has been read. During the first sweep, the
\emph{Apply}\footnote{Similar to \cite{Sanghavi1996}, we refer to all top-down manipulating sweeps
  as Apply, e.g.\ \texttt{not}, \texttt{or}, and \texttt{if-then-else}. This even includes
  \texttt{identity} which merely reverses the edges.}, the entire recursion tree is unfolded
top-down. Here, the priority queues also double as a computation cache~\cite{Brace1990,Minato1990}
by merging separate paths to the same recursion target. Hence, the resulting output is in fact not a
tree but a DAG. Yet, it is only an OBDD and needs to be reduced. To do so, Adiar uses an
I/O-efficient variant of the original bottom-up Reduce algorithm by
Bryant~\cite{Bryant1986,Arge1995:2}. Here, a priority queue is used to forward the \texttt{uid} of
reduced nodes $t'$ in the final ROBDD to their to be reduced parents $s$ in the intermediate OBDD.
Yet, to know the parents $s$, the Reduce needs the intermediate OBDD to be transposed, i.e.\ the
DAG's edges to be reversed. Luckily, the Apply sweep outputs its arcs (directed edges) sorted by
their target\footnote{In \cite{Soelvsten2022:TACAS}, arcs to terminals are actually output to a
  separate file sorted by their source. This is merely to improve performance. Hence, we will ignore
  this detail.}. This effectively transposes the OBDD and so no extra work is needed for the Reduce
\cite{Arge1995:1,Soelvsten2022:TACAS}. For example, the \texttt{or} of
\cref{fig:bdd_example:x1,fig:bdd_example:x0&-x1} creates the unreduced BDD in
\cref{fig:or:example:apply:terms} with the arc-based representation in
\cref{fig:bdd_example:representation:arc}.

The I/O and time complexity of this Apply--Reduce tandem is
\if\arxiv1%
  \begin{equation*}
    \Oh{\sort{N + T}} \enspace,
  \end{equation*}
\else%
  $\Oh{\sort{N + T}}$, %
\fi%
where $N$ is the size of the input(s) and $T$ is the size of the unreduced output of the Apply sweep
\cite{Soelvsten2022:TACAS}. To catch up with conventional implementation's performance, major %
efforts have been dedicated to improve on this foundation.

\subsubsection*{Levelised Cuts \cite{Soelvsten2023:ATVA}}

The arcs placed in the above-mentioned priority queues correspond to graph cuts in the (R)OBDDs. In
particular, these cuts correspond to the border between the processed and to be processed nodes.
Since the priority queues induce a levelised processing order, the shape of these cuts match the
(R)OBDD's levels.

This ensures that the maximum size of the priority queues is bounded by the
maximum levelised cut in the input. Hence, sound upper bounds on these cuts can in
turn be used to determine a priori whether a faster internal-memory priority queue can be used.


In practice, this improves performance for smaller and moderate instances.

\subsubsection*{Levelised Random Access \cite{Soelvsten2024:SPIN}}

Orthogonally, an Apply sweep for a product construction, e.g.\ an \texttt{or}, can be simplified if one
of its inputs is narrow, i.e.\ each level fits into internal memory. In this case, one can load each
level in its entirety into internal memory. This provides random access to all of its nodes on
said level, making one of two priority queues within the Apply (as in~\cite{Soelvsten2022:TACAS}) obsolete.

In practice, this improves performance for larger instances.

\section{I/O-efficient Multi-variable Quantification} \label{sec:theory}

\if\arxiv1
  \begin{figure}[b]
\else
  \begin{figure}[b]
\fi
  \centering

  \begin{tikzpicture}[every text node part/.style={align=center}]
    \draw (0,0) rectangle ++(2,1)
    node[pos=.5]{\texttt{Apply}\\{\small $\PQ{outer}{\downarrow}$}};
    \draw (3.5,0) rectangle ++(2,1)
    node[pos=.5]{\texttt{Reduce}\\{\small $\PQ{outer}{\uparrow}$}};

    \draw[->] (-0.5,0.5) -- ++(0.5,0)
    node[pos=-0.55]{$f$};

    \draw[->, dashed] (2,0.5) -- ++(1.5,0)
    node[pos=0.5, above]{\small $\FILE{outer}$};

    \draw[->] (5.5,0.5) -- ++(3,0)
    node[pos=0.22, above]{\small $\FILE{outer}'$}
    node[pos=1.23]{$\exists \vec{x} : f$};

    \draw (7,-2) rectangle ++(2,1)
    node[pos=.5]{\texttt{Apply}\\{\small $\PQ{inner}{\downarrow}$}};
    \draw (3.5,-2) rectangle ++(2,1)
    node[pos=.5]{\texttt{Reduce}\\{\small $\PQ{inner}{\uparrow}$}};

    \draw[->] (7.5,0.5) -- ++(0,-1.5);

    \draw[->, dashed] (7,-1.5) -- ++(-1.5,0)
    node[pos=0.5, above]{\small $\FILE{inner}$};

    \draw[->] (4.5,-1) -- ++(0,1)
    node[pos=0.5, right]{\small $\FILE{inner}'$};
  \end{tikzpicture}

  \caption{The Apply--Reduce pipeline of \texttt{exists} with Nested Sweeping. $F_{*}$ are files
    created by the respective sweep while $\PQ{*}{*}$ are the priority queues they each use.}
  \label{fig:nested_tandem}
\end{figure}

\if\arxiv1
  \begin{figure}[t]
\else
  \begin{figure}[t]
\fi
  \centering

  \begin{tikzpicture}[every text node part/.style={align=center}]
    \node at (-0.5,-2.0) {$x_j$};
    \node at (-0.5,-1.0) {$x_i$};

    \draw (0,-1.0) rectangle ++(1.6,1)
    node[pos=.5]{\texttt{Apply}\\{\small $\PQ{outer}{\downarrow}$}};

    \draw[->, dashed] (0.8,-1.1) -- ++(0,-2.4);

    \draw (1.2,-4.0) rectangle ++(1.6,1)
    node[pos=.5]{\texttt{Reduce}\\{\small $\PQ{outer}{\uparrow}$}};

    \draw[->] (2.0,-2.9) -- ++(0,0.8);
    \draw[->] (2.0,-2.0) -- ++(0,0.9);
    \draw[->] (2.0,-1.0) -- ++(0,0.9);

    \draw (3.4,-2.6) rectangle ++(1.6,1)
    node[pos=.5]{\texttt{Apply}\\{\small $\PQ{inner}{\downarrow}$}};

    \draw[->, dashed] (4.2,-2.7) -- ++(0,-0.8);

    \draw (4.6,-4.0) rectangle ++(1.6,1)
    node[pos=.5]{\texttt{Reduce}\\{\small $\PQ{inner}{\uparrow}$}};

    \draw[->] (5.4,-2.9) -- ++(0,0.9);

    \draw (6.8,-1.6) rectangle ++(1.6,1)
    node[pos=.5]{\texttt{Apply}\\{\small $\PQ{inner}{\downarrow}$}};

    \draw[->, dashed] (7.6,-1.7) -- ++(0,-1.8);

    \draw (8.0,-4.0) rectangle ++(1.6,1)
    node[pos=.5]{\texttt{Reduce}\\{\small $\PQ{inner}{\uparrow}$}};

    \draw[->] (8.8,-2.9) -- ++(0,1.9);

    \draw[-, densely dotted] (0.9,-3.5) -- ++(0.3,0);
    \draw[-, densely dotted] (2.1,-2.1) -- ++(1.3,0);
    \draw[-, densely dotted] (4.3,-3.5) -- ++(0.3,0);
    \draw[-, densely dotted] (5.4,-1.9) -- ++(0,0.4) -- ++(-2.1,0) -- ++(0,-0.5) -- ++(-1.3,0);
    \draw[-, densely dotted] (2.1,-1.1) -- ++(4.7,0);
    \draw[-, densely dotted] (7.7,-3.5) -- ++(0.3,0);
    \draw[-, densely dotted] (8.8,-0.9) -- ++(0,0.4) -- ++(-2.1,0) -- ++(0,-0.5) -- ++(-4.7,0);
  \end{tikzpicture}

  \caption{Sweep direction (solid/dashed) and control-flow (dotted) of Nested Sweeping. The y-axis
    corresponds to the levels within the BDD. Variables $x_i$ and $x_j$ induce nested sweeps; for
    \texttt{exists}, they are to be quantified variables.}
  \label{fig:nested_sweeps}
\end{figure}

\if\arxiv1
  \begin{figure}[t]
\else
  \begin{figure}[!b]
\fi
  \centering

  \begin{tikzpicture}
    \draw (-5.5,0) -- (-4.5,4) -- (-3.5,0) -- cycle;

    \node[inner sep=0pt] at (-4.5,3.3) (l1) {$\circ$};

    \node[inner sep=0pt] at (-4.6,2.5) (l2) {$\circ$};

    \draw[->] (l1) -- (l2);

    \node[inner sep=0pt] at (-4.3,1.0) (l3) {$\circ$};

    \draw[->] (l1) edge[bend left=10] (l3);
    \draw[->] (l2) edge[bend right=20] (l3);

    \draw[->] (l3) -- ++(-0.2,-0.5);
    \draw[->] (l3) -- ++( 0.2,-0.5);

    \node[inner sep=0pt] at (-4.9,1.0) (l4) {$\circ$};

    \draw[->] (l2) edge[bend right=10] (l4);

    \draw[->] (l4) -- ++(-0.2,-0.5);
    \draw[->] (l4) -- ++( 0.2,-0.5);

    \node at (-3.25,2) {$\odot$};

    \draw (-3,0) -- (-2,4) -- (-1,0) -- cycle;

    \node[inner sep=0pt] at (-2,3.3) (r1) {$\circ$};

    \node[inner sep=0pt] at (-2.25,1.8) (r2) {$\circ$};

    \draw[->] (r1) -- (r2);

    \node[inner sep=0pt] at (-1.65,1.8) (r3) {$\circ$};

    \draw[->] (r1) -- (r3);

    \node[inner sep=0pt] at (-2.5,1.0) (r4) {$\circ$};

    \draw[->] (r2) -- (r4);

    \draw[->] (r4) -- ++(-0.2,-0.5);
    \draw[->] (r4) -- ++( 0.2,-0.5);

    \node[inner sep=0pt] at (-2.0,1.0) (r5) {$\circ$};

    \draw[->] (r2) -- (r5);
    \draw[->] (r3) -- (r5);

    \draw[->] (r5) -- ++(-0.2,-0.5);
    \draw[->] (r5) -- ++( 0.2,-0.5);

    \node[inner sep=0pt] at (-1.5,1.0) (r6) {$\circ$};

    \draw[->] (r3) -- (r6);

    \draw[->] (r6) -- ++(-0.2,-0.5);
    \draw[->] (r6) -- ++( 0.2,-0.5);

    \node at (-0.5,2) {$\Rightarrow$};

    \input{tikz/nested_cases__tree.tex}

    \node[inner sep=0pt] at (3.0,3.5)  (n1) {$\circ$};

    \node[inner sep=0pt] at (2.3,2.3)  (n2) {$\circ$};
    \draw[<-] (n1) -- (n2);

    \node[inner sep=0pt] at (1.9,1.0) (n3) {$\circ$};
    \draw[<-] (n2) -- (n3);

    \draw[<-] (n3) -- ++(-0.3,-0.5);
    \draw[<-] (n3) -- ++(0.3,-0.5);

    \node[inner sep=0pt] at (3.2,1.68) (n4) {$\circ$};
    \draw[<-] (n2) -- (n4);

    \draw[<-] (n4) -- ++(-0.3,-0.5);

    \node[inner sep=0pt] at (4.3,1.68) (n5) {$\circ$};
    \draw[<-] (n1) -- (n5);

    \node[inner sep=0pt] at (3.6,1.0) (n6) {$\circ$};
    \draw[<-] (n5) -- (n6);
    \draw[<-] (n4) -- (n6);

    \draw[<-] (n6) -- ++(-0.3,-0.5);
    \draw[<-] (n6) -- ++(0.3,-0.5);

    \node[inner sep=0pt] at (4.6,1.0) (n7) {$\circ$};
    \draw[<-] (n5) -- (n7);

    \draw[<-] (n7) -- ++(-0.3,-0.5);
    \draw[<-] (n7) -- ++(0.3,-0.5);
  \end{tikzpicture}

  \caption{Outer Apply one (or more) input BDD(s) are processed in a top-down sweep to create a
    single transposed BDD.}
  \label{fig:nested_cases:transpose}
\end{figure}

\input{fig/nested_cases__example.tex}

Our previous work in \cite{Soelvsten2022:TACAS} only covers simple BDD operations without any
data-dependencies in its recursion, e.g.\ the \texttt{or}. Yet, this does not cover the
\texttt{exists} in \cref{fig:exists}, where the nested call to \texttt{or} on line 8 depends on the
recursions from lines 4 and 5.

To address this, we introduce the \emph{nested sweeping} framework. As shown in
\cref{fig:nested_tandem,fig:nested_sweeps}, nested sweeping wraps the algorithm(s) depicted in
\cref{fig:tandem}: after transposing the input in an initial Apply sweep, a single \emph{outer}
Reduce sweep accumulates the result of multiple \emph{inner} Apply--Reduce sweeps. For
\texttt{exists}, the inner Apply is an \texttt{or} sweep.

More precisely, nested sweeping consists of the four phases described below and depicted
step-by-step in \cref{fig:nested_cases:example}.

\paragraph{Outer Apply:}

As shown in \cref{fig:nested_cases:transpose}, inputs are combined (and possibly manipulated) in an
Apply sweep into a single file, \FILE{outer}. This transposes and merges the inputs such that they
are of the form needed by the Reduce of \cite{Soelvsten2022:TACAS}.

In the case of \texttt{exists}, this is merely a simple transposition of $f$. In
\if\arxiv1%
  \cref{sec:theory:quantify optimisations}%
\else%
  the full paper~\cite{Soelvsten2024:arXiv}%
\fi%
, we explore ways in which this phase can also do double duty for partially resolving the
quantification computation.
\if\arxiv1%
  In the case of the relational product, the conjunction of states and relation can be computed as
  part of this phase.

\fi%

\paragraph{Outer Reduce:}

As in \cite{Soelvsten2022:TACAS}, each level of \FILE{outer} is reduced bottom-up by having a
priority queue, \PQ{outer}{\uparrow}, forward the information about reduced nodes, $t'$, to their
unreduced parents, $s$. The reduced output is pushed into a new file, $\FILE{outer}'$.

Let $x_j$ be the next level that needs a nested sweep. For \texttt{exists}, $x_j$ is the largest
still to be quantified variable in $X$. As visualised in \cref{fig:nested_cases:outer}, the logic
of \cite{Soelvsten2022:TACAS} is extended as follows:
\begin{enumerate}
\item \label{nested:outer:recurse}

  If the current level is $x_j$, each arc $\arc{s}{t'}$ to a reduced node $t'$ at this level is
  turned into a request and placed in a second priority queue,
  \PQ{inner}{\downarrow}.

  For \texttt{exists}, the requests are of the form $\arc{s}{(\lowof{t'},\highof{t'})}$.

\item \label{nested:outer:reduce}

  If the current level is deeper than $x_j$, nodes are reduced as in \cite{Soelvsten2022:TACAS}
  with a caveat: whether the arc $\arc{s}{t'}$ to the reduced node $t'$ is placed
  in \PQ{outer}{\uparrow} or in \PQ{inner}{\downarrow} depends on the level of the unreduced
  parent $s$ as follows:
  \begin{enumerate}
  \item \label{nested:outer:reduce:outer}

    If $x_j \leq \topof{s}$, i.e.\ $s$ is as deep or deeper than level $x_j$, then
    $\arc{s}{t'}$ is placed in $\PQ{outer}{\uparrow}$ as normal.

  \item \label{nested:outer:reduce:inner}

    Otherwise, i.e.\ if $\topof{s} < x_j$, $\arc{s}{t'}$ is placed in \PQ{inner}{\downarrow}
    instead.
  \end{enumerate}
\end{enumerate}
For \texttt{exists}, Case~\ref{nested:outer:recurse} matches the invocation of \texttt{or} on line
8 of \cref{fig:exists} whereas \ref{nested:outer:reduce} is the return with an unquantified
variable on line 7.

When level $x_j$ has finished processing, \PQ{inner}{\downarrow} is populated with all requests that
span across level $x_j$. Now, the inner Apply sweep is invoked.

\if\arxiv1
  \begin{figure}[t]
\else
  \begin{figure}[t]
\fi
  \centering

  \begin{tikzpicture}
    \input{tikz/nested_cases__tree.tex}
    \input{tikz/nested_cases__xj.tex}

    \node[inner sep=1pt] at (2.8,0.4) (t1) {$t'$};

    \node[inner sep=1pt] at (2.4,1.3) (s1) {$s$};
    \draw[<-, thick, solid] (t1) -- (s1);

    \node[inner sep=1pt] at (3.3,1.67) (s2) {$s$};
    \draw[<-, thick, solid] (t1) -- (s2);

    \node[inner sep=1pt] at (2.3,2.3) (s3) {$s$};

    \node[inner sep=1pt] at (1.8,1.0) (t3) {$t'$};
    \draw[<-, thick, dotted] (t3) -- (s3);

    \node[inner sep=1pt] at (3.0,3.4) (s4) {$s$};
    \node[inner sep=1pt, red] at (4.3,1.68) (t4) {$t'$};
    \draw[<-, thick, dotted] (t4) -- (s4);

    \node[inner sep=0pt] at (3.6,1.0) (t4-1) {$\circ$};
    \draw[-, dotted, thick] (t4-1) -- (t4);

    \draw[->, solid] (t4-1) -- ++(0.3,-0.5);
    \draw[->, solid] (t4-1) -- ++(-0.3,-0.5);

    \node[inner sep=0pt] at (4.6,1.0) (t4-2) {$\circ$};
    \draw[-, dotted, thick] (t4-2) -- (t4);

    \draw[->, solid] (t4-2) -- ++(0.3,-0.5);
    \draw[->, solid] (t4-2) -- ++(-0.3,-0.5);
  \end{tikzpicture}

  \caption{Outer Reduce: solid arcs stay in \PQ{outer}{\uparrow}
    (Case~\ref{nested:outer:reduce:outer}) while dotted arcs are turned into requests for
    \PQ{inner}{\downarrow} (Cases \ref{nested:outer:reduce:inner} on the left and
    \ref{nested:outer:recurse} on the right). }
  \label{fig:nested_cases:outer}
\end{figure}

\if\arxiv1
  \begin{figure}[t]
\else
  \begin{figure}[t]
\fi
  \centering

  \subfloat[Inner Apply: starting with multiple root requests (dotted), $\FILE{inner}$ is
  constructed with the to-be preserved subtrees (left) together with new nodes that are products of
  previous ones (right).] {
    \label{fig:nested_cases:inner:apply}

    \begin{tikzpicture}[scale=0.9]
      \input{tikz/nested_cases__tree.tex}
      \input{tikz/nested_cases__xj.tex}
      \input{tikz/nested_cases__xi.tex}

      \draw[very thick, dotted, pattern=horizontal lines, pattern color=black!12!white]
      (0,0) -- (1.13,1.5) -- (4.87,1.5) -- (6,0) -- cycle;

      \node[inner sep=0pt] at (2.3,2.3) (s1) {$\circ$};

      \node[draw=black, shape=circle, inner sep=0pt] at (1.9,1.0) (t1) {$\circ$};
      \draw[<-, thick, dotted] (s1) -- (t1);

      \draw[<-, solid] (t1) -- ++(0.3,-0.5);
      \draw[<-, solid] (t1) -- ++(-0.3,-0.5);

      \node[inner sep=0pt] at (3.0,3.5) (s2) {$\circ$};

      \node[inner sep=0pt] at (4.3,1.68) (t2) {$\cdot$};
      \draw[<-, thick, dotted] (s2) -- (t2);

      \node[draw=black, shape=circle, inner sep=0pt] at (4.3,1.0) (t2-1) {$\circ \circ$};
      \draw[<-, solid] (t2-1) -- ++(0.3,-0.5);
      \draw[<-, solid] (t2-1) -- ++(-0.3,-0.5);

      \draw[-, thick, dotted] (t2) -- (t2-1);
    \end{tikzpicture}
  }%
  \quad%
  \subfloat[Inner Reduce: arcs below $x_j$ stay within the inner sweep
  (Case~\ref{nested:inner:current}, solid). Arcs that cross $x_j$ are given back to the outer
  (Case~\ref{nested:inner:parent}, dotted) or to the next inner sweep (Case~\ref{nested:inner:next},
  dash dotted).] {
    \label{fig:nested_cases:inner:reduce}

    \begin{tikzpicture}[scale=0.9]
      \input{tikz/nested_cases__tree.tex}
      \input{tikz/nested_cases__xj.tex}
      \input{tikz/nested_cases__xi.tex}

      \draw[very thick, dotted, pattern=horizontal lines, pattern color=black!12!white]
      (0,0) -- (1.13,1.5) -- (4.87,1.5) -- (6,0) -- cycle;

      \node at (2.3,2.3) (s2) {$s$};

      \node at (1.9,1.0) (t2) {$t''$};
      \draw[<-, thick, dotted] (t2) -- (s2);

      \node at (2.6,1.2) (s1) {$s$};

      \node at (3.5,0.4) (t1) {$t''$};
      \draw[<-, thick, solid] (t1) -- (s1);

      \node at (3.0,3.5) (s3) {$s$};

      \node at (4.5,1.0) (t3) {$t''$};
      \draw[<-, thick, loosely dash dot] (t3) -- (s3);
    \end{tikzpicture}
  }

  \caption{Visualization of the Inner Apply and the Inner Reduce.}
  \label{fig:nested_cases:inner}
\end{figure}

\paragraph{Inner Apply:}

As depicted in \cref{fig:nested_cases:inner:apply}, starting with the requests in
\PQ{inner}{\downarrow}, the reduced nodes, $t'$, placed in $\FILE{outer}'$ by the outer Reduce
sweep, are processed with an Apply sweep from \cite{Soelvsten2022:TACAS}. The intermediate unreduced
result is placed in a new file, \FILE{inner}.

For \texttt{exists}, this sweep is the execution of the \texttt{or} on line 8 of \cref{fig:exists}.
While the algorithms in \cite{Soelvsten2022:TACAS} were only applied to a single BDD, they can also
be applied to an entire forest; this merely requires prepopulating their priority queue with
recursion requests to each root. Hence, we can reuse the previous top-down algorithms from
\cite{Soelvsten2022:TACAS} as is.

\paragraph{Inner Reduce:}

After the inner Apply sweep, \FILE{inner} is reduced in another bottom-up Reduce sweep of
\cite{Soelvsten2022:TACAS}. This creates the reduced nodes $t''$ placed in a new file
$\FILE{inner}'$. Let $x_i$ be the next level above $x_j$ that also needs a nested sweep. For
\texttt{exists}, $x_i$ is the largest variable smaller than $x_j$ that also needs to be quantified.
The arc $\arc{s}{t''}$ is placed in a priority queue, \PQ{inner}{\uparrow}, as follows.
\begin{enumerate}
\item \label{nested:inner:current}

  If $x_j < \topof{s}$, i.e.\ the parent $s$ is below level $x_j$, then
  $\arc{s}{t''}$ is forwarded within this inner Reduce sweep's priority queue,
  $\PQ{inner}{\uparrow}$.

\item \label{nested:inner:parent}

  If $\topof{s} \in [x_i, x_j]$, i.e.\ the parent $s$ is between level $x_i$ and $x_j$ then
  $\arc{s}{t''}$ is given back to the outer sweep, $\PQ{outer}{\uparrow}$. This matches
  Case~\ref{nested:outer:reduce:outer} in the outer Reduce.

\item \label{nested:inner:next}

  If $\topof{s} < x_i$, i.e. the parent $s$ is above $x_i$, then $\arc{s}{t''}$ is placed into
  $\PQ{inner}{\downarrow}$ to prepare the next invocation of an inner Apply sweep. This matches
  Case~\ref{nested:outer:reduce:inner} in the outer Reduce.
\end{enumerate}
The three cases above are depicted in \cref{fig:nested_cases:inner:reduce}. For \texttt{exists},
Cases~\ref{nested:inner:parent} and \ref{nested:inner:next} are equivalent to the return from
\texttt{or} back to \texttt{exists}. Case~\ref{nested:inner:current} is equivalent to a return
statement within the \texttt{or}'s own recursion. Case~\ref{nested:inner:next} is needed to match
\ref{nested:outer:reduce:inner} with $x_j$ replaced with $x_i$.

Finally, $\FILE{inner}'$ replaces $\FILE{outer}'$ and control returns to the outer Reduce sweep to
proceed with the levels above $x_j$.

\bigskip

\noindent As shown in \cref{fig:nested_cases:example}, $\FILE{inner}$ from the inner Apply sweep can
be thought of as overlayed on top of $\FILE{outer}$ from the outer Apply sweep; together they
produce a valid (but unreduced) OBDD. The outer and inner Reduce sweeps work together to reduce this
into a single file, $\FILE{outer}'$. When no more levels, $x_j$, need to be processed and the outer
Reduce sweep has finished processing, then $\FILE{outer}'$ contains the final reduced BDD of all
nested operations.

As mentioned above, the inner Apply sweep needs its priority queue $\PQ{inner}{\downarrow}$ to be
prepopulated with all relevant roots. This is done in Cases~\ref{nested:outer:recurse} and
\ref{nested:outer:reduce:inner} in the outer and Case~\ref{nested:inner:next} in the inner Reduce
sweeps.

Since the result of the inner Apply and Reduce sweeps replaces the entire set of nodes in
$\FILE{outer}'$, the priority queue $\PQ{inner}{\downarrow}$ not only needs to be populated with
requests for the nodes that need to be changed but also with requests for the nodes one wishes to
keep (see also \cref{fig:nested_cases:inner:apply,fig:nested_cases:example:outer reduce:x1}). This
makes the inner Apply sweep not only compute the desired result but also act as a mark-and-sweep
garbage collection. On the first glance, these additional non-modifying requests may seem too costly
-- especially if most requests do not modify subtrees. In practice, $33.3\%$ of all requests created
throughout our benchmarks in \cref{sec:experiments} are subtree modifying. For each benchmark
instance, $23.0\%$ of all requests modify subtrees on average (median $35.6\%$). That is, a
reasonable number of all requests (and hence BDD nodes processed) change the subgraph in
$\FILE{outer}'$.

\subsection{Complexity of Nested Sweeping} \label{sec:theory:complexity}

As mentioned in the description of the outer Apply sweep, nested sweeping works for multiple inputs.
In this work, it suffices to assume it only has to deal with a single BDD $f$ of $N$ nodes as also
depicted in \cref{fig:nested_tandem,fig:nested_cases:example}.

\begin{lemma} \label{thm:transpose:complexity}

  A single BDD $f$ with $N$ nodes can be transposed in $\Theta(\sort{N})$ I/Os and time and
  $\Theta(N)$ space.
\end{lemma}
\begin{proof}
  In $\Theta(\scan{N})$ I/Os and time iterate over and split all nodes $v$ in-order into the two
  arcs $\arc[dashed]{\uidof{v}}{\lowof{v}}$ and $\arc[solid]{\uidof{v}}{\highof{v}}$. Sort these
  $2N$ arcs on their target using $\Theta(\sort{N})$ I/Os and time and linear space to transposes
  them.
\end{proof}

\if\arxiv1%

  In \cref{sec:theory:quantify optimisations}, we propose to embed valuable computations inside of
  the outer Apply sweep. This comes at the cost of potentially changing the BDD size. To encapsulate
  such cases too, let $\Nprime$ be the output size of the outer Apply sweep which may exceed
  $\Oh{N}$. Yet, this step is not the bottle-neck of the algorithm.

\fi%

\begin{lemma} \label{thm:outer:complexity}

  Ignoring the work done within the inner sweeps, the outer Reduce sweep costs
  $\Theta(\sort{\Nprime})$ I/Os and time and requires $\Oh{\Nprime}$ space.
\end{lemma}
\begin{proof}
  This follows from the complexity of the Reduce algorithm in \cite{Soelvsten2022:TACAS} (based on
  \cite{Arge1996}) and the constant extra time spent for each of the $\Nprime$ nodes to resolve the
  additional logic in Cases~\ref{nested:outer:recurse} and \ref{nested:outer:reduce} of the outer
  Reduce sweep.
\end{proof}

By combining \cref{thm:transpose:complexity,thm:outer:complexity} together with the fact that the
last invocation of the inner Apply and Reduce sweeps constructs, together with the outer Reduce
sweep, the output of size $T$, we obtain the following lower bound on the complexity of nested
sweeping.

\begin{corollary} \label{thm:nested:complexity-lower}

  Nested Sweeping uses $\Omega(N+T)$ space and $\Omega(\sort{N+T})$ time and I/Os where $N$ and $T$
  are the size of the input and output, respectively.
\end{corollary}

In particular for the \texttt{exists} BDD operation, let $T_{j}$ be the size of $\FILE{outer}'$ when
the inner Apply sweep is invoked at level $x_j$.

\begin{lemma} \label{thm:inner or:complexity}

  A single invocation of the inner Apply and Reduce sweeps at $x_j$ costs
  $\Theta(\sort{\Nprime + T_{j}^2})$ I/Os and time and uses $O(\Nprime + T_{j}^2)$ space.
\end{lemma}
\begin{proof}
  As in \cite{Soelvsten2022:TACAS}, the algorithm's complexity depends on the number of elements
  placed in the priority queues \cite{Arge1995:2}. In particular, a single nested \texttt{or} sweep
  deals with up to $2 \Nprime$ arcs from $\FILE{outer}$. On top of these, it also processes up to
  $2 T_{j}^2$ arcs created during the product construction of $\FILE{outer}'$.
\end{proof}

Since nested sweeping closely simulates the (parallelised) recursive BDD algorithm in
\cref{fig:exists}, one should expect it achieves, similar to the algorithms in
\cite{Soelvsten2022:TACAS}, major improvements in the number of I/Os at the cost of a $\log$-factor
in the running time when compared to the conventional recursive algorithms. This is indeed the case.

\begin{proposition} \label{thm:multi quantification:complexity-upper}

  Quantification of a set of variables, $X$, is computable with nested sweeping in
  $O(\sort{N^{2^{\lvert X \rvert}}})$ I/Os and time and $O(N^{2^{\lvert X \rvert}})$ space.
\end{proposition}
\begin{proof}
  Due to \cref{thm:transpose:complexity}, $\FILE{outer}$ from the outer Apply sweep has up to $2N$
  arcs. This is also the size of $\FILE{outer}'$ without any inner sweeps. Each inner Apply and
  Reduce sweep may increase the size of $\FILE{outer}'$ quadratically. The result follows from
  \cref{thm:inner or:complexity,thm:transpose:complexity,thm:outer:complexity}.
\end{proof}

Asymptotically, this is not an improvement over just quantifying each variable one-by-one using the
algorithm already proposed in the full version of \cite{Soelvsten2022:TACAS}. Yet, doing so would
involve $2 \lvert X \rvert$ sweeps over \emph{all} levels of the input whereas, as highlighted in
\cref{fig:nested_sweeps,fig:nested_cases:example}, nested sweeping only processes levels below each
quantified variable. This difference is also evident in practice: throughout our benchmarks in
\cref{sec:experiments}, when quantifying with nested sweeping rather than each variable
independently, the total number of requests processed with the \texttt{or} operation decreases by
$13.9\%$ while the share of 2-ary product constructions increases from $57.3\%$ to $66.6\%$.

%

\subsection{Optimisations for Nested Sweeping} \label{sec:theory:nested optimisations}

While nested sweeping as described above is an improvement over previous work in
\cite{Soelvsten2022:TACAS}, there are multiple avenues to further improve its performance in
practice.

\if\arxiv1
\subsubsection{Terminal Arcs:}

No inner Apply sweep changes the value of terminals. Hence, requests of the form $\arc{s}{\top}$ and
$\arc{s}{\bot}$ can be forwarded to $s$ regardless of any nesting levels, $x_j$, in-between.
Furthermore, the request based on $t'$ in the outer Reduce sweep may trivially resolve into a
terminal. In this case, the resulting terminal can be forwarded to its parents (regardless of their
level). For \texttt{exists}, this would be if both $\lowof{t'}$ and $\highof{t'}$ are terminals or
either of them is the $\top$ terminal.

This decreases the size of \PQ{inner}{\downarrow}. Furthermore, it makes the requests placed in
\PQ{inner}{\downarrow} compatible with an implicit invariant of the Apply sweeps in
\cite{Soelvsten2022:TACAS}.

In practice, the number of requests skipped this way depends on the use-case and the scale. $3.7\%$
of all requests processed as part of our benchmarks (see \cref{sec:experiments} for a description)
are for terminals. On average, $6.9\%$ of the requests (with a median of $7.0\%$) are for terminals
in each benchmark.
For the Garden of Eden (GoE) benchmark specifically, $15.3\%$ of all requests are terminals on
average (median $17.7\%$). On the other hand for the Quantified Boolean Formula (QBF) benchmark,
only $5.0\%$ (median $5.9\%$) of them were.

\fi

\subsubsection{Bail-out of Inner Sweep:}

There is no need for the outer Reduce sweep to invoke the inner sweeps if \PQ{inner}{\downarrow}
only contains requests that preserve subtrees, i.e.\ if Case~\ref{nested:outer:recurse} in the outer
Reduce did not create any requests that manipulate the accumulated OBDD in $\FILE{outer}'$. On level
$x_j$, such requests can stem from a redundant node $t'$ being suppressed. For \texttt{exists}, this
may also occur due to either $\lowof{t'}$ or $\highof{t'}$ being the $\bot$ terminal, which is
neutral for the \texttt{or} operation, or being $\top$, which is short-circuiting it.

In this case, the entire content of \PQ{inner}{\downarrow} can be redistributed between
\PQ{outer}{\uparrow} and \PQ{inner}{\downarrow} for the next deepest to be quantified level, $x_i$.
After doing so, the outer Reduce sweep can immediately proceed processing the next level.

For \texttt{exists}, any short-circuiting by the \texttt{or} operation in
Case~\ref{nested:outer:recurse} of the outer Reduce sweep can kill off some subtrees in
$\FILE{outer}'$. In this case, one cannot skip the last invocation of the inner sweeps. Otherwise,
the final result $\FILE{outer}'$ could include dead nodes. Yet, even so, one can instead of the
expensive top-down algorithm, e.g.\ \texttt{or} for \texttt{exists}, invoke the inner Apply sweep
with a much simpler (and therefore faster) mark-and-sweep algorithm.

In practice, $75.6\%$ of all nested sweeps in our benchmarks (see \cref{sec:experiments} for their
presentation) are skippable. For each benchmark, between $6.8\%$ and $93.5\%$ of all nested
computations were skipped with an average of $59.2\%$ (median of $81.0\%$). The number of nested
levels depends on the problem domain and its instance. For the Garden of Eden (GoE) benchmark, only
$26.8\%$ of all nested computations were skipped on average (median $29.0\%$), whereas $82.7\%$
(median $84.3\%$) of all levels of the Quantified Boolean Formulas (QBFs) could be skipped.

\subsubsection{Root Requests Sorter:}

Instead of making the outer Reduce sweep push requests directly into \PQ{inner}{\downarrow}, it can
push it into an intermediate list of requests, \LIST{outer}{\downarrow}. The content of
\LIST{outer}{\downarrow} is sorted using the same ordering as \PQ{inner}{\downarrow} as the inner
Apply sweep is invoked, to then merge it on the fly with the inner Apply sweep's priority queue. This
allows one to postpone initialising this priority queue until the inner Apply sweep is invoked. This
has multiple benefits:
\begin{itemize}

\item \PQ{inner}{\downarrow}, resp.\ \PQ{inner}{\uparrow}, only exists and uses internal memory
  during the inner Apply sweep, resp.\ the inner Reduce sweep. Hence, the memory otherwise dedicated
  to \PQ{inner}{\downarrow} can be used in the outer Reduce sweep for the Reduce's per-level data
  structures in \cite{Soelvsten2022:TACAS}. Furthermore, this also increases the amount of space
  available to the inner Reduce sweep. Hence, this ought to improve the running time of both the
  inner and the outer Reduce sweeps.

\if\arxiv1%
\item In practice, sorting a list of elements once is significantly faster than maintaining an order
  in a priority queue \cite{Molhave2012}. Merging \LIST{outer}{\downarrow} on the fly with
  \PQ{inner}{\downarrow} is faster than passing requests to the inner sweep's priority queue.%
\fi%

\item If \PQ{inner}{\downarrow}, resp.\ \PQ{inner}{\uparrow}, is initialised for each inner Apply
  sweep, resp.\ inner Reduce sweep, then the monotonic and faster \emph{levelised priority queue} in
  the full version of \cite{Soelvsten2022:TACAS} can be used instead of a regular non-monotonic
  priority queue.

\item Levelised cuts \cite{Soelvsten2023:ATVA} bound the size of the priority queues in each
  individual inner Apply and Reduce sweep. Hence, for each nested sweep, one can, if it is safe to
  do so, replace \PQ{inner}{\downarrow} and/or \PQ{inner}{\uparrow} with a faster internal memory
  variant.

\item Levelised random access \cite{Soelvsten2024:SPIN} changes the order in which nodes are
  resolved on each level. This has to be accommodated for in the sorting predicate in
  \PQ{inner}{\downarrow}. Hence, \LIST{outer}{\downarrow} allows for this optimisation to be applied
  for each invocation of the inner Apply sweep depending on the width of $\FILE{outer}'$.

\end{itemize}
Furthermore, levelised cuts not only bound the size of \PQ{outer}{\uparrow} in the outer Reduce
sweep but also the size of \LIST{outer}{\downarrow}. Hence, while deciding whether
\PQ{outer}{\uparrow} fits into memory, one can also decide whether \LIST{outer}{\downarrow} does.

All in all, this allows the optimisations in \cite{Soelvsten2023:ATVA,Soelvsten2024:SPIN} to be
applied on a sweep-by-sweep basis. In practice, if one neither uses faster internal-memory variants
of \PQ{inner}{\downarrow} and \PQ{inner}{\uparrow} nor levelised random access, then Adiar needs a
total of 32.1~h to solve 145 out of the 147 benchmarks in \cref{sec:experiments}. Using these
two optimisations shaves 13.0~h off the total computation time (speedup of $1.68$). For each
individual instance, this improves Adiar's performance between a factor of $1.07$ and $5.05$ ($1.77$
on average)%
\if\arxiv1%
\footnote{%
  Compared to \cite{Soelvsten2023:ATVA}, the \emph{external} memory sorters in this comparison still
  use the levelised cuts to circumvent wasting time with initialising too much internal memory. This
  is why, there is not a speedup of several orders of magnitude.
  If this use of levelised cuts is also reverted to obtain its state back in
  \cite{Soelvsten2022:TACAS}, then preliminary experiments on a machine with 8 GiB of memory
  exhibits a speedup of $2.71$ on average.
  As memory increases, one should expect a difference similar to the one reported in
  \cite{Soelvsten2023:ATVA}}%
\fi%
. Furthermore, without \LIST{outer}{\downarrow}, the exponential blow-up in \cref{thm:multi
  quantification:complexity-upper} implies \PQ{inner}{\downarrow} would almost always have to use
external memory. As the optimisations in
\cite{Soelvsten2022:TACAS,Soelvsten2023:ATVA,Soelvsten2024:SPIN} would then not be applicable, one
would expect a slowdown of several orders of magnitude similar to
\cite{Soelvsten2023:ATVA}.

\if\arxiv1%
  \subsection{Optimisations for Quantification} \label{sec:theory:quantify optimisations}

  As presented above, the outer Apply sweep merely transposes the input BDD $f$. Yet, doing so may
  not make the most out of having to touch the entire BDD graph; as long as the result is a
  transposed BDD, one can incorporate additional computations inside of this sweep. Hence, we now
  explore possible top-down sweeps that can be used instead of the algorithm in
  \cref{thm:transpose:complexity}.

  \subsubsection{Pruning $\top$ Siblings:}

  \Cref{thm:nested:complexity-lower} and \cref{thm:multi quantification:complexity-upper} show a
  possibly wide gap in the potential performance of the nested \texttt{exists} algorithm.
  \Cref{thm:inner or:complexity} shows this stems from the possibility of some partially quantified
  result explodes exponentially in size. Yet, $T_{j}$ can only be larger than $T$ if it contains
  subtrees that will be pruned or merged later when another variable is quantified. This can only
  happen due to the $\top$ terminal shortcircuiting an \texttt{or}. Hence, to be closer to the lower
  bound in \cref{thm:nested:complexity-lower}, we need to identify redundant computation by pushing
  information about the $\top$ terminal down through the BDD of $f$.

  \begin{figure}[t]
  \centering

  \subfloat[Pruning due to $\top$.] {
    \label{fig:pruning_quantification:top}
    \centering

    \begin{tikzpicture}
      \node at (0, 1) (r) {};

      \node[shape = circle,    draw = black] at (0, 0) (n) {$x_i$};

      \node at (-0.6, -1) (alpha) {$\alpha$};
      \node[shape = rectangle, draw = black] at ( 0.6, -1) (T) {$\top$};

      \draw[->, dashed]
        (n) edge (alpha)
      ;
      \draw[->]
        (r) edge (n)
        (n) edge (T)
      ;

      \node at (1.3, 0) {$\Rightarrow$};

      \node at (2.4, 1) (r) {};

      \node[shape = rectangle, draw = black] at ( 2.4, -1) (T) {$\top$};

      \draw[<-]
        (r) edge (T)
      ;
    \end{tikzpicture}
  }
  \quad
  \subfloat[Skipping node due $\bot$.] {
    \label{fig:pruning_quantification:bot}
    \centering

    \begin{tikzpicture}
      \node at (0, 1) (r) {};

      \node[shape = circle,    draw = black] at (0, 0) (n) {$x_i$};

      \node at (-0.6, -1) (alpha) {$\alpha$};
      \node[shape = rectangle, draw = black] at ( 0.6, -1) (F) {$\bot$};

      \draw[->, dashed]
        (n) edge (alpha)
      ;
      \draw[->]
        (r) edge (n)
        (n) edge (F)
      ;

      \node at (1.3, 0) {$\Rightarrow$};

      \node at (2.4, 1) (r) {};

      \node at ( 2.4, -1) (alpha) {$\alpha$};

      \draw[<-]
        (r) edge (alpha)
      ;
    \end{tikzpicture}
  }

  \caption{Example of pruning quantification of a to-be quantified level $x_i$.}
  \label{fig:pruning_quantification}
\end{figure}

  As shown in \cref{fig:pruning_quantification:top}, one can collapse to be quantified nodes at a
  level $x_i$ if one of their children is the $\top$ terminal. Similarly, as shown in
  \cref{fig:pruning_quantification:bot}, one can skip over nodes with a $\bot$ terminal as its
  child. This can be done as part of a simple top-down sweep similar to the Restrict in the full
  version of \cite{Soelvsten2022:TACAS}.

  In the worst-case, this does not apply to any node in $f$ and so the output is similar to the
  algorithm in \cref{thm:transpose:complexity}. Our preliminary experiments indicate this
  approach introduces an overhead of up to $2\%$. Yet, if nodes are prunable, making $\Nprime < N$,
  then total performance can improve with up to $21\%$.

  \subsubsection{Deepest Variable Quantification:}

  The single-variable quantification in the full version of \cite{Soelvsten2022:TACAS} can also be
  used to transpose $f$. This removes one of the to be quantified variables $x_i \in X$ in
  $\Oh{\sort{N^2}}$ I/Os and time and $\Oh{N^2}$ space. The resulting transposed graph,
  $\FILE{outer}$, has size $\Nprime \leq N^2$. To not change the overall memory usage, one can
  choose $x_i$ to be the largest to be quantified variable. Doing so makes the levels at $x_i$ and
  below equivalent to $\FILE{inner}$ after the first inner Apply and Reduce sweep. That is,
  $\Nprime \leq N + T_{i}$ and one saves an entire nested sweep at no cost to memory usage.

  Our preliminary experiments indicate this only slows down computation time on average by $4.7\%$.
  We hypothesise this is due to the deepest variable $x_i$ is often close to the bottom of the BDD
  and so, this sweep is primarily transposing the graph with more complex logic.

  One can also incorporate the above pruning of $\top$ siblings inside this quantification sweep.
  This improves performance for applicable cases. But, it does not offset the additional overhead
  in the remaining cases.

  \subsubsection{Partial Quantification:}

  The single-variable quantification in the full version of \cite{Soelvsten2022:TACAS} can be
  generalised to partially resolve all $x_i \in X$ in a single sweep. The requests in
  \cite{Soelvsten2022:TACAS} are for pairs of nodes $(t_1,t_2) \in f \times f$.

  \begin{figure}[!b]
  \centering

  \subfloat[Fully quantified pair of nodes] {
    \label{fig:partial_quantification:skip}
    \centering

    \begin{tikzpicture}
      \node at (0, 1) (r) {};

      \node at (-0.6, 1) (r1) {};
      \node[shape = circle,    draw = black] at (-0.6, 0) (n1) {$x_i$};

      \node at ( 0.6, 1) (r2) {};
      \node[shape = circle,    draw = black] at ( 0.6, 0) (n2) {$x_i$};

      \node[shape = rectangle, draw = black] at (-1.2, -1) (F) {$\bot$};
      \node at ( 0.0, -1) (alpha) {$\alpha$};
      \node at (1.2, -1) (beta) {$\beta$};

      \draw[->, dashed]
        (n1) edge (F)
        (n2) edge (alpha)
      ;
      \draw[->]
        (r1) edge (n1)
        (n1) edge (alpha)
        (r2) edge (n2)
        (n2) edge (beta)
      ;

      \node at (1.8, 0) {$\Rightarrow$};

      \node at (2.6, 1) (r) {};

      \node at (2.6, -1) (prod) {$\alpha \times \beta$};

      \draw[<-]
        (r) edge (prod)
      ;
    \end{tikzpicture}
  }
  \quad
  \subfloat[Partially quantified pair of nodes] {
    \label{fig:partial_quantification:node}
    \centering

    \begin{tikzpicture}
      \node at (0, 1) (r) {};

      \node at ( 0.0, 1) (r1) {};
      \node[shape = circle,    draw = black] at ( 0.0, 0) (n1) {$x_i$};

      \node at (-0.6, -1) (alpha) {$\alpha$};
      \node[shape = rectangle, draw = black] at ( 0.6, -1) (F) {$\bot$};

      \node at ( 2.2, 1) (r2) {};
      \node[shape = circle,    draw = black] at ( 2.2, 0) (n2) {$x_i$};

      \node at ( 1.6, -1) (gamma) {$\gamma$};
      \node at (2.8, -1) (beta) {$\beta$};

      \draw[->, dashed]
        (n1) edge (alpha)
        (n2) edge (gamma)
      ;
      \draw[->]
        (r1) edge (n1)
        (n1) edge (F)
        (r2) edge (n2)
        (n2) edge (beta)
      ;

      \node at (3.4, 0) {$\Rightarrow$};

      \node at ( 4.8, 1) (r) {};

      \node[shape = circle,    draw = black] at ( 4.8, 0) (nprod) {$x_i$};

      \node at ( 4.2, -1) (prod1) {$\alpha \times \beta$};
      \node at ( 5.4, -1) (prod2) {$\gamma$};

      \draw[<-, dashed]
        (nprod) edge (prod1)
      ;
      \draw[<-]
        (r) edge (nprod)
        (nprod) edge (prod2)
      ;
    \end{tikzpicture}
  }

  \caption{Example of partial quantification of a level $x_i$.}
  \label{fig:partial_quantification}
\end{figure}

  Without loss of generality, assume $\topof{t_1} = \topof{t_2} = x_i \in X$. In this case, the
  2-ary product construction should turn into
  $(\lowof{t_1}, \highof{t_1}, \lowof{t_2}, \highof{t_2})$. If any of these four \texttt{uid}s are
  the $\top$ terminal then the entire request can immediately be resolved to $\top$. Furthermore,
  any $\bot$ terminal is neutral to the \texttt{or} operation and can be pruned from the 4-tuple.
  Similarly, any duplicate \texttt{uid}s can be merged. Since \texttt{or} is commutative, one can
  sort the 4-tuple to quickly identify these cases. If the resulting tuple has 2 or fewer entries
  remaining, the product construction can proceed as in the full version of
  \cite{Soelvsten2022:TACAS} (\cref{fig:partial_quantification:skip}). Otherwise, a new node with
  $x_i$ is created with the result of each half of the 4-tuple as its children
  (\cref{fig:partial_quantification:node}). Inductively, this is correct (after later quantification
  of the new $x_i$ node and its subtrees) as the \texttt{or} operation is associative.

  The resulting DAG is a 2-ary product construction of $f$, and so \FILE{outer} has size
  $\Nprime \leq N^2$. As per \cite{Soelvsten2022:TACAS}, this single sweep is computable in
  $\Oh{\sort{N^2}}$ I/Os and time and $\Oh{N^2}$ space.

  Similar to the two $\top$ terminal pruning above, partial quantification prunes shortcutted
  subtrees across all levels of $f$. Furthermore, similar to deepest quantification, it leaves at
  least one fewer levels of to be quantified variables to be processed later.

  Our preliminary experiments indicate, partial quantification can in practice improve performance
  up to $61\%$. Yet, many other intances slow down just as much (up to $115\%$, i.e.\ also a bit
  more than a factor of two). We hypothesise this is due to partial quantification pairing nodes
  with a conflicting assignment. For example, in \cref{fig:partial_quantification:node} $\alpha$ is
  paired with $\beta$ rather than $\gamma$. Oddly enough, in our preliminary experiments, the
  instances that were improved by $\top$ pruning are disjoint from the ones improved by partial
  quantification. Further research is needed to investigate why and when partial quantification is
  useful.

  \subsubsection{Repeated Partial Quantification:}

  The top-down sweep aboves produces a transposed and unreduced OBDD. Yet, it is in practice
  possible that $\Nprime$ is smaller than $(1+\epsilon) \cdot N$ for some $\epsilon > -1$,
  i.e.\ its size has not grown considerable. In this case, the resulting product construction has
  very few new BDD nodes that are potentially reducible. Hence, it may be more beneficial to
  untranspose the OBDD and then immediately rerun another transposing top-down sweep. Doing so with
  pruning or partial quantification can propagate the $\top$ terminal further and so prune more
  subtrees. Yet, it is unlikely that pruning $\top$ terminals in
  \cref{fig:pruning_quantification:top} makes said terminal available for another to be quantified
  variable. That is, it is unlikely in this case that a second sweep would further prune subtrees.
  Hence, this is most promising to do with partial quantification.

  Since there are very few new BDD nodes, it is unlikely that the Reduce sweep of
  \cite{Soelvsten2022:TACAS} will do much more than just untranspose the DAG. Hence, one would want
  to untranspose it with a simpler and faster algorithm. As can be seen in
  \cref{fig:bdd_example:representation}, one can instead merely sort all arcs on their source and
  then merge them on the fly into nodes. Asymptotically, this is still a $\Theta(\sort{\Nprime})$
  operation. But, the constant involved is smaller than the Reduce of \cite{Soelvsten2022:TACAS}.

  Hence, one can repeat the above partial quantification operation until $\Nprime$ exceeds
  $(1+\epsilon) \cdot N$, it has run $\delta$ times, or there are no more to be quantified variables
  left in \FILE{outer}. In practice, we have not yet found any instance where more than a single
  quantification sweep further improves performance. If further research can uncover when it is
  beneficial to use partial quantification, then we might be able to extrapolate this into a value
  of $\delta$.

\fi%

\if\arxiv1
  \section{Implementation of Nested Sweeping in Adiar} \label{sec:implementation}

  Most of the logic in \cref{sec:theory} can be implemented by wrapping the priority queues
  \PQ{outer}{\uparrow}, \PQ{inner}{\downarrow}, and \PQ{inner}{\uparrow} and
  \LIST{outer}{\downarrow} with additional logic on how to merge and whereto split requests.
  \begin{itemize}
  \item The logic of whether to push to \PQ{outer}{\uparrow} or \LIST{outer}{\downarrow} in the
    outer Reduce sweep, resp.\ Cases~\ref{nested:inner:parent} and \ref{nested:inner:next} in the
    inner Reduce sweep, is a conditional on level $x_j$, resp.\ $x_i$.

  \item During the inner Apply sweep, requests from \LIST{outer}{\downarrow} are merged on the fly
    with the ones pushed to \PQ{inner}{\downarrow}.

  \item When placing requests in \LIST{outer}{\downarrow}, they are marked as originating from the
    outer Reduce sweep. During the inner Reduce sweep, requests are forwarded to
    \PQ{outer}{\uparrow} or \PQ{inner}{\uparrow} depending on whether they are marked to be from the
    outer sweep or not.
  \end{itemize}
  This has been implemented in Adiar v2.0 with (compile-time known) \emph{decorators}: a class with
  the same interface as the priority queues runs the above logic before passing it onto the wrapped
  priority queues and sorters. This makes the logic of each sweep agnostic to and reusable in the
  context of nested sweeping.

The nested sweeping framework, i.e.\ the decorators, \LIST{outer}{\downarrow}, and the algorithm and
its optimisations, has been implemented with 1287 lines of templated C++ classes and functions.
Similar to \cite{Soelvsten2023:NFM,Soelvsten2023:ATVA}, the use of templates completely remove any
indirection and abstraction introduced for the sake of code quality. The entire framework has been
tested separately from the remaining codebase with 104 unit tests. The quantification algorithms
themselves grew from 548 lines of code and 84 unit tests to 1152 lines of code and 152 unit tests
(without any of the optimisations in \cref{sec:theory:quantify optimisations}).

\fi

\section{Experimental Evaluation} \label{sec:experiments}

\if\arxiv0%
  We have implemented the nested sweeping framework in \cref{sec:theory} in Adiar~v2.0. To evaluate
  the impact of this addition, we have run experiments aiming at answering the following three
  research questions:

\else%
  To evaluate the impact of adding the nested sweeping framework, we have run experiments aiming at
  answering the following three research questions:

\fi%

\begin{enumerate}

\item\label{rq:improvement} How does nested sweeping compare to the repeated use of the
  single-variable quantification from the full version of \cite{Soelvsten2022:TACAS}?

\item\label{rq:cal} How does Adiar with nested sweeping compare to the external memory BDD package,
  CAL~\cite{Sanghavi1996}?

\item\label{rq:competitors} How does Adiar with nested sweeping compare to conventional BDD packages
  that use depth-first recursion and memoisation
  \cite{Lind1999,Somenzi2015,Benes2020,Husung2024,Dijk2016}?

\end{enumerate}

\noindent As we are only concerned with the design of BDD algorithms, not when and how to use them,
we will not compare to non-BDD based approaches.

\subsection{Benchmarks}

We have implemented the following two benchmarks that rely on multi-variable quantification. Similar
to \cite{Soelvsten2022:TACAS,Soelvsten2023:NFM}, all benchmarks have been implemented on top of
C++-templated adapters for each BDD package. This makes each BDD package run the exact same set of
operations without introducing any indirection. The source code for all benchmarks can be found at
the following url:
\begin{center}
  \href{https://github.com/ssoelvsten/bdd-benchmark}{github.com/ssoelvsten/bdd-benchmark}
\end{center}

\subsubsection{QBF Solving:}

Given a Quantified Boolean Formula (QBF) in the \texttt{QCIR}~\cite{QBF2014} format, each gate of
the given circuit is recursively transformed into a BDD. For inputs, we use the 102 encodings from
\cite{Shaik2023:arXiv} of 2-player games on a grid. In our experience, the symbolic style of these
inputs makes them well suited to be solved with BDDs. Hence, they provide a typical use-case of
quantification in BDDs. Furthermore, though these inputs are not in CNF they are in prenex form. In
practice, resolving these prenex quantifications at the end is computationally much more expensive
than computing the to be quantified circuit, i.e.\ the \emph{matrix}.

Based on preliminary experiments, we use a variable order based on a depth-first traversal of the
given circuit. In the prenex, we merge adjacent blocks with the same quantifier to increase the
number of concurrently quantified variables.

\subsubsection{Garden-of-Eden:}

A \emph{Garden-of-Eden}~\cite{Moore1962} (GoE) is any configuration of a cellular automaton without
a predecessor. For Conway's Game of Life on a grid of size $n_r \times n_c$, this can be encoded as a
relation with $(n_r+2) \cdot (n_c+2)$ \emph{previous}-state variables, $\vec{x}$, and
$n_r \cdot n_c$ \emph{next}-state variables, $\vec{x'}$. Next-state variables, $x_i' \in \vec{x'}$,
in the BDD follow a row-major order. Previous-state variables, $x_i \in \vec{x}$, are interleaved to
directly precede their respective next-state variable, $x_i'$. A GoE corresponds to an unsatisfying
assignment to $x_i' \in \vec{x'}$ after having quantified all previous-state variables, $\vec{x}$.

Recent results show there exists no GoE for Conway's Game of Life of size $8 \times 8$ or smaller
\cite{Beer2022}. Hence, the BDD for the $n_r \times n_c \leq 8 \times 8$ sized transition relation
will collapse to $\top$ after existential quantification. Yet, the row-major encoding of the
transition relation itself only requires a polynomially sized BDD. Hence, the complexity of this
problem manifests during the existential quantification.

\subsection{Hardware and Settings}

As in \cite{Soelvsten2022:TACAS,Soelvsten2023:NFM,Soelvsten2023:ATVA,Soelvsten2024:SPIN}, we have
run our experiments on the \emph{Grendel} cluster at the Centre for Scientific Computing Aarhus. In
particular, we ran both benchmarks on machines with 48-core $3.0$ GHz Intel Xeon Gold 6248R
processors, $384$ GiB of RAM, $3.5$ TiB of SSD disk, and run Rocky Linux (Kernel 4.18.0-513). All
code was compiled with GCC $10.1.0$ or \texttt{rustc} $1.72.1$.
Each BDD package was given $\tfrac{9}{10}$th of the available RAM, i.e.\ $345$~GiB, leaving
$\tfrac{1}{10}$th to other data structures and the operating system. Next to that, the BDD packages
use a single thread and their default/recommended settings.

Note that these machines have vasts amounts of memory. This is to ensure that depth-first
implementations are not slowed down by external factors. If less memory is available, then
depth-first implementations would have to run multiple garbage collections to stay within the memory
limits (cf.\ the largest instances solved by BuDDy~\cite{Lind1999} in
\cref{fig:nested_vs_depth-first:buddy}). This, in turn, clears their memoisation tables and forces
them to recompute previous results. Furthermore, this large amount of memory ensures they can solve
larger problems without using the swap partition. If they had to use it then they would slow down by
about two orders of magnitude (see the full paper of \cite{Soelvsten2022:TACAS} for an example).
Hence, machines of this scale allow us to measure the algorithms' running time without the noise
otherwise introduced by their execution environment. Finally, this biases the running time in favour
of the depth-first implementations, which in turn makes the numbers we report on Adiar's relative
performance close to the worst-case.

Adiar does not yet support dynamic variable reordering. Hence, we have disabled reordering for all
other BDD packages as part of these experiments.

\subsection{Experimental Results} \label{sec:experiments:results}

The computing cluster's scheduler does not let many long-running jobs run concurrently. To obtain
all 1176 data points reported below within only a few months, we had to place each of the 147
instances in buckets of instances with a common timeout. In particular, an instance is placed in the
bucket with the smallest timeout that is four times larger than Adiar needed during preliminary
experiments. That is, a BDD package timing out should only be understood as it (possibly) being
considerably slower than Adiar.

Depending on an instance bucket placement, running time measurements were made 1 to 3 times. Due to
node failures on the cluster, Adiar with nested sweeping was run once more, resulting in its
measurements being repeated on many instances 4 times. On average, all data points had $3.0$
measurements. Similar to \cite{Soelvsten2022:TACAS,Soelvsten2023:NFM,Soelvsten2023:ATVA}, we report
for each benchmark the minimum time recorded as it is the measurement with the least noise
\cite{Chen2016}. Average ratios are aggregated using the geometric mean.

Adiar needs less than 1~s to solve 27 out of the 45 GoE instances, resp.\ 50 out of the 102 QBF
instances. As will become evident later with \cref{fig:nested_vs_cal,fig:nested_vs_depth-first},
this makes them so small that they are not within the current scope of Adiar. For completeness, we
still show and discuss these results.

\if\arxiv1
  \begin{figure}[!b]
\else
  \begin{figure}[t]
\fi
  \centering

  \begin{tikzpicture}
    \begin{axis}[%
      width=0.92\linewidth, height=0.45\linewidth,
      every tick label/.append style={font=\scriptsize},
      xlabel={\scriptsize Running Time of Adiar with Single-variable Quantification (s)},
      xmin=0.01,
      xtick={0.01,0.1,1,10,100,1000,10000,100000},
      xmax=100000,
      xmode=log,
      ylabel={\scriptsize Speed-up with Nested Sweeping},
      ymin=0.125,
      ymax=2,
      ymode=log,
      ytick={0.125, 0.25, 0.5, 1, 2},
      yticklabels={
        $8 \times$,
        $4 \times$,
        $2 \times$,
        $1 \times$,
        $\tfrac{1}{2} \times$
      },
      y dir=reverse,
      grid style={dashed,black!12},
      ]

      \addplot[domain=0.01:1000000, samples=8, color=black]
      {1};

      \draw[densely dotted, opacity=0.4] (0,0.5) -- (0,-1.9);

      \begin{scope}[blend mode=soft light]

        \addplot+ [forget plot, style=dots_goe]
        table {./data/singleton_goe.tex};
        \addplot+ [forget plot, style=dots_qbf]
        table {./data/singleton_qbf.tex};
      \end{scope}
    \end{axis}
  \end{tikzpicture}

  {
    \tikzdot{goe} GoE \quad \tikzdot{qbf} QBF
  }

  \caption{Relative time ($T_{\text{old}} / T_{\text{new}}$) of quantification with nested sweeping
    ($T_{\text{new}}$) compared to the previous repeated single-variable quantification
    ($T_{\text{old}}$).
  }
  \label{fig:nested_vs_singleton}
\end{figure}

\subsubsection{RQ~\ref{rq:improvement}: Improvement by Nested Sweeping}
\label{sec:experiments:singleton}

\Cref{fig:nested_vs_singleton} shows the speed-up of using Adiar with nested sweeping %
\if\arxiv1%
  (without any optimisations in \cref{sec:theory:quantify optimisations}) %
\fi%
relative to quantifying each variable individually. Across all instance sizes, nested sweeping is in
general an improvement in performance. We have recorded a slowdown of up a factor of $1.05$ for 5
instances. Yet, we also recorded speed-ups up to a factor of $5.88$ for the 142 remaining instances.
On average, performance improves by a factor of $1.7$ for both QBF and GoE. The total computation
time was decreased by $21\%$ from $49.4$~h to $39.1$~h.

\subsubsection{RQ~\ref{rq:cal}: Comparison to CAL}
\label{sec:experiments:cal}

To the best of our knowledge, CAL~\cite{Sanghavi1996} (based on \cite{Ochi1993,Ashar1994}) is the
only other BDD package also designed to manipulate BDDs larger than main memory. To do so, it uses
breadth-first algorithms that should work well with BDDs stored on disk via the operating system's
swap memory \cite{Sanghavi1996}. CAL also includes algorithms to support multi-variable
quantification \cite{Sanghavi1996}. For more details, see \cite{Sanghavi1996} and \cref{sec:related
  work}. The machines for our experiments provide a 48~GiB swap partition, i.e.\ a $12.5\%$ increase
in available space.

\if\arxiv1
  \begin{table}[t!]
\else
  \begin{table}[t]
\fi
  \centering

  \caption{Time taken and the average ratio between Adiar and CAL for the 124 commonly solved
    instances. The average (geometric mean) pertains only to instances where Adiar needed at least
    1~s to solve them. Ratios larger than $1.00$ means Adiar is faster.}
  \label{tab:cal}

  \bgroup
  \def\arraystretch{1.1}
  \setlength\tabcolsep{5pt}
  \begin{tabular}{l||r|r||c|c||c|c}
            & \multicolumn{2}{c||}{Time}
            & \multicolumn{2}{c||}{\# Solved}
            & \multicolumn{2}{c}{Avg. Ratio (1+s)}
    \\
            & \multicolumn{1}{c|}{\tikzdot{goe} GoE} & \multicolumn{1}{c||}{\tikzdot{qbf} QBF}
            & \tikzdot{goe} GoE                      & \tikzdot{qbf} QBF
            & \tikzdot{goe} GoE                      & \tikzdot{qbf} QBF
    \\ \hline \hline
    Adiar   & 7431.9s                                & 688.0s
            & 45                                     & 102
            & --                                     & --
    \\ \hline
    CAL     & 184688.3s                              & 295660.0s
            & 38                                     & 86
            & 5.0                                    & 25.2
  \end{tabular}
  \egroup
\end{table}

\if\arxiv1
  \begin{figure}[t]
\else
  \begin{figure}[t]
\fi
  \centering

  \begin{tikzpicture}
    \begin{axis}[%
      width=0.92\linewidth, height=0.45\linewidth,
      every tick label/.append style={font=\scriptsize},
      xlabel={\scriptsize Running Time of Adiar with Nested Sweeping (s)},
      xmin=0.01,
      xtick={0.01,0.1,1,10,100,1000,10000,100000},
      xmax=100000,
      xmode=log,
      ylabel={\scriptsize Relative Time of CAL},
      ymin=0.0039,
      ymax=4000,
      ytick = {0.00390625,0.015625,0.0625,0.25,1,4,16,64,256,1024},
      yticklabels = {
        $2^{-8} \times$,
        $2^{-6} \times$,
        $2^{-4} \times$,
        $2^{-2} \times$,
        $1 \times$,
        $2^{2} \times$,
        $2^{4} \times$,
        $2^{6} \times$,
        $2^{8} \times$,
        $2^{10 \times}$
      },
      ymode=log,
      grid style={dashed,black!12},
      ]

      \addplot[domain=0.01:1000000, samples=8, color=black]
      {1};

      \draw[densely dotted, opacity=0.4] (0,7.0) -- (0,-4.5);

      \begin{scope}[blend mode=soft light]


        \addplot+ [forget plot, style=dots_goe]
        table {./data/cal_goe__solved.tex};
        \addplot+ [forget plot, style=dots_qbf]
        table {./data/cal_qbf__solved.tex};

        \addplot+ [forget plot, style=x_goe]
        table {./data/cal_goe__timeouts.tex};
        \addplot+ [forget plot, style=x_qbf]
        table {./data/cal_qbf__timeouts__truncated.tex};
      \end{scope}
    \end{axis}
  \end{tikzpicture}

  {
    \tikzdot{goe} GoE \quad \tikzdot{qbf} QBF
  }

  \caption{Relative performance ($T_{\text{CAL}} / T_{\text{Adiar}}$) of CAL ($T_{\text{CAL}}$)
    compared to Adiar with Nested Sweeping ($T_{\text{Adiar}}$). Time-/Memouts are marked as crosses.
  }
  \label{fig:nested_vs_cal}
\end{figure}

Preliminary experiments indicated CAL's breadth-first algorithms are much slower than Adiar's
time-forward processing. Hence, we multiplied the timeout for CAL by a factor of $3$. But as is
evident in \cref{fig:nested_vs_cal}, this increase turned out to still overestimate CAL's
performance on larger instances. Hence, the running times and averages in \cref{fig:nested_vs_cal}
and \cref{tab:cal} pertain only to the 124 instances which CAL can solve within the given RAM, SWAP,
and the time limits.

Even though this discards the instances where CAL struggles, i.e.\ the data points that remain are
in CAL's favour, \cref{fig:nested_vs_cal} shows Adiar heavily outperforms CAL for instances where
Adiar takes $1$~s or longer to solve. Where CAL uses $133.4$~h to solve 124 instances, Adiar, by
solving them in only $2.3$~h, is $59.1$ times faster. On these larger instances, CAL is on average
$14.7$ times slower than Adiar. As is evident in \cref{fig:nested_vs_cal} and \cref{tab:cal}, Adiar
especially outperforms CAL on the QBF benchmark. For example, the largest difference was measured
for the \texttt{hex/hein\_15\_5x5-13} QBF instance, where CAL is $1081$ times slower than the
$71.2$~s Adiar needs to solve it.

CAL is considerably faster for the instances where Adiar takes less than $1$~s to solve. At this
small scale, CAL uses the conventional depth-first algorithms \cite{Soelvsten2023:ATVA} to sidestep
the performance issues of its external memory breadth-first algorithms. Doing the same for Adiar is
still left as future work \cite{Soelvsten2023:ATVA,Soelvsten2024:SPIN}.

\if\arxiv1
  \begin{figure}[htbp!]
\else
  \begin{figure}[p]
\fi
  \centering

  \subfloat[BuDDy] {
    \label{fig:nested_vs_depth-first:buddy}

    \begin{tikzpicture}
      \begin{axis}[%
        width=0.45\linewidth, height=0.41\linewidth,
        every tick label/.append style={font=\scriptsize},
        xlabel={\scriptsize Adiar Running Time (s)},
        xmin=0.01,
        xmax=100000,
        xtick={0.01,0.1,1,10,100,1000,10000,100000},
        xmode=log,
        ylabel={\scriptsize Relative Running Time},
        ymin=0.00390625,
        ymax=5.0,
        ytick = {0.00390625,0.015625,0.0625,0.25,1,4},
        yticklabels = {
          $2^{-8} \times$,
          $2^{-6} \times$,
          $2^{-4} \times$,
          $2^{-2} \times$,
          $1 \times$,
          $2^{2} \times$
        },
        ymode=log,
        grid style={dashed,black!12},
        ]

        \addplot[domain=0.001:100000, samples=8, color=black]
        {1};

        \draw[densely dotted, opacity=0.4] (0,1.0) -- (0,-5.0);

        \begin{scope}[blend mode=soft light]

          \addplot+ [forget plot, style=dots_goe] table {./data/buddy_goe__solved.tex};
          \addplot+ [forget plot, style=dots_qbf] table {./data/buddy_qbf__solved.tex};

          \addplot+ [forget plot, style=x_goe] table {./data/buddy_goe__timeouts__truncated.tex};
          \addplot+ [forget plot, style=x_qbf] table {./data/buddy_qbf__timeouts__truncated.tex};
        \end{scope}
      \end{axis}
    \end{tikzpicture}
  }
  \subfloat[CUDD] {
    \label{fig:nested_vs_depth-first:cudd}

    \begin{tikzpicture}
      \begin{axis}[%
        width=0.45\linewidth, height=0.41\linewidth,
        every tick label/.append style={font=\scriptsize},
        xlabel={\scriptsize Adiar Running Time (s)},
        xmin=0.01,
        xmax=100000,
        xtick={0.01,0.1,1,10,100,1000,10000,100000},
        xmode=log,
        ylabel={\scriptsize Relative Running Time},
        ymin=0.00390625,
        ymax=5.0,
        ytick = {0.00390625,0.015625,0.0625,0.25,1,4},
        yticklabels = {
          $2^{-8} \times$,
          $2^{-6} \times$,
          $2^{-4} \times$,
          $2^{-2} \times$,
          $1 \times$,
          $2^{2} \times$
        },
        ymode=log,
        grid style={dashed,black!12},
        ]

        \addplot[domain=0.01:100000, samples=8, color=black]
        {1};

        \draw[densely dotted, opacity=0.4] (0,1.0) -- (0,-5.0);

        \begin{scope}[blend mode=soft light]

          \addplot+ [forget plot, style=dots_goe] table {./data/cudd_goe__solved.tex};
          \addplot+ [forget plot, style=dots_qbf] table {./data/cudd_qbf__solved.tex};

          \addplot+ [forget plot, style=x_goe] table {./data/cudd_goe__timeouts__truncated.tex};
          \addplot+ [forget plot, style=x_qbf] table {./data/cudd_qbf__timeouts__truncated.tex};
        \end{scope}
      \end{axis}
    \end{tikzpicture}
  }

  \medskip

  \subfloat[LibBDD] {
    \label{fig:nested_vs_depth-first:libbdd}

    \begin{tikzpicture}
      \begin{axis}[%
        width=0.45\linewidth, height=0.41\linewidth,
        every tick label/.append style={font=\scriptsize},
        xlabel={\scriptsize Adiar Running Time (s)},
        xmin=0.01,
        xmax=100000,
        xtick={0.01,0.1,1,10,100,1000,10000,100000},
        xmode=log,
        ylabel={\scriptsize Relative Running Time},
        ymin=0.00390625,
        ymax=5.0,
        ytick = {0.00390625,0.015625,0.0625,0.25,1,4},
        yticklabels = {
          $2^{-8} \times$,
          $2^{-6} \times$,
          $2^{-4} \times$,
          $2^{-2} \times$,
          $1 \times$,
          $2^{2} \times$
        },
        ymode=log,
        grid style={dashed,black!12},
        ]

        \addplot[domain=0.01:100000, samples=8, color=black]
        {1};

        \draw[densely dotted, opacity=0.4] (0,1.0) -- (0,-5.0);

        \begin{scope}[blend mode=soft light]

          \addplot+ [forget plot, style=dots_goe] table {./data/libbdd_goe__solved.tex};
          \addplot+ [forget plot, style=dots_qbf] table {./data/libbdd_qbf__solved.tex};

          \addplot+ [forget plot, style=x_goe] table {./data/libbdd_goe__timeouts__truncated.tex};
          \addplot+ [forget plot, style=x_qbf] table {./data/libbdd_qbf__timeouts__truncated.tex};
        \end{scope}
      \end{axis}
    \end{tikzpicture}
  }

  \medskip

  \subfloat[OxiDD] {
    \label{fig:nested_vs_depth-first:oxidd}

    \begin{tikzpicture}
      \begin{axis}[%
        width=0.45\linewidth, height=0.41\linewidth,
        every tick label/.append style={font=\scriptsize},
        xlabel={\scriptsize Adiar Running Time (s)},
        xmin=0.01,
        xmax=100000,
        xtick={0.01,0.1,1,10,100,1000,10000,100000},
        xmode=log,
        ylabel={\scriptsize Relative Running Time},
        ymin=0.00390625,
        ymax=5.0,
        ytick = {0.00390625,0.015625,0.0625,0.25,1,4},
        yticklabels = {
          $2^{-8} \times$,
          $2^{-6} \times$,
          $2^{-4} \times$,
          $2^{-2} \times$,
          $1 \times$,
          $2^{2} \times$
        },
        ymode=log,
        grid style={dashed,black!12},
        ]

        \addplot[domain=0.01:100000, samples=8, color=black]
        {1};

        \draw[densely dotted, opacity=0.4] (0,1.0) -- (0,-5.0);

        \begin{scope}[blend mode=soft light]

          \addplot+ [forget plot, style=dots_goe] table {./data/oxidd_goe__solved.tex};
          \addplot+ [forget plot, style=dots_qbf] table {./data/oxidd_qbf__solved.tex};

          \addplot+ [forget plot, style=x_goe] table {./data/oxidd_goe__timeouts__truncated.tex};
          \addplot+ [forget plot, style=x_qbf] table {./data/oxidd_qbf__timeouts__truncated.tex};
        \end{scope}
      \end{axis}
    \end{tikzpicture}
  }
  \subfloat[Sylvan] {
    \label{fig:nested_vs_depth-first:sylvan}

    \begin{tikzpicture}
      \begin{axis}[%
        width=0.45\linewidth, height=0.41\linewidth,
        every tick label/.append style={font=\scriptsize},
        xlabel={\scriptsize Adiar Running Time (s)},
        xmin=0.01,
        xmax=100000,
        xtick={0.01,0.1,1,10,100,1000,10000,100000},
        xmode=log,
        ylabel={\scriptsize Relative Running Time},
        ymin=0.00390625,
        ymax=5.0,
        ytick = {0.00390625,0.015625,0.0625,0.25,1,4},
        yticklabels = {
          $2^{-8} \times$,
          $2^{-6} \times$,
          $2^{-4} \times$,
          $2^{-2} \times$,
          $1 \times$,
          $2^{2} \times$
        },
        ymode=log,
        grid style={dashed,black!12},
        ]

        \addplot[domain=0.01:100000, samples=8, color=black]
        {1};

        \draw[densely dotted, opacity=0.4] (0,1.0) -- (0,-5.0);

        \begin{scope}[blend mode=soft light]

          \addplot+ [forget plot, style=dots_goe] table {./data/sylvan_goe__solved.tex};
          \addplot+ [forget plot, style=dots_qbf] table {./data/sylvan_qbf.tex};

          \addplot+ [forget plot, style=x_goe] table {./data/sylvan_goe__timeouts__truncated.tex};
        \end{scope}
      \end{axis}
    \end{tikzpicture}
  }

  \vspace{7pt}

  {
    \tikzdot{goe} GoE \quad \tikzdot{qbf} QBF
  }

  \caption{Relative performance of depth-first implementations compared to Adiar with nested
    sweeping. Time-/Memouts are marked as crosses. 
  }
  \label{fig:nested_vs_depth-first}
\end{figure}

\label{sec:experiments:depth-first}

\subsubsection{RQ~\ref{rq:competitors}: Comparison to Depth-First Implementations}

\if\arxiv1
  \begin{table}[!b]
\else
  \begin{table}[!b]
\fi
  \centering

  \caption{Total time needed by Adiar and conventional depth-first implementations to solve the 140
    commonly solved instances. The average (geometric mean) covers all instances that were commonly
    solved by all BDD packages and where Adiar needed at least 1~s to solve. Ratios smaller than
    $1.00$ means Adiar is slower.}
  \label{tab:depth-first}

  \medskip

  \bgroup
  \def\arraystretch{1.1}
  \setlength\tabcolsep{5pt}

  \begin{tabular}{l||r|r||c|c||c|c}
            & \multicolumn{2}{c||}{Time}
            & \multicolumn{2}{c||}{\# Solved}
            & \multicolumn{2}{c}{Avg. Ratio (1+s)}
    \\
            & \multicolumn{1}{c|}{\tikzdot{goe} GoE} & \multicolumn{1}{c||}{\tikzdot{qbf} QBF}
            & \tikzdot{goe} GoE                      & \tikzdot{qbf} QBF
            & \tikzdot{goe} GoE                      & \tikzdot{qbf} QBF
    \\ \hline \hline
    Adiar   & 9655.7s                                & 4499.4s
            & 45                                     & 102
            & --                                     & --
    \\ \hline
    BuDDy   & 4725.5s                                & 3793.1s
            & 40                                     & 100
            & 0.30                                   & 0.25
    \\
    CUDD    & 10892.8s                               & 4591.9s
            & 40                                     & 101
            & 0.61                                   & 0.75
    \\
    LibBDD  & 4365.7s                                & 2687.3s
            & 43                                     & 101
            & 0.54                                   & 0.45
    \\
    OxiDD   & 21223.9s                               & 2379.6s
            & 41                                     & 101
            & 0.48                                   & 0.39
    \\
    Sylvan  & 2925.4s                                & 5841.4s
            & 44                                     & 102
            & 0.46                                   & 0.70
  \end{tabular}
  \egroup
\end{table}

We have compared the performance of Adiar with BuDDy~2.4~\cite{Lind1999},
CUDD~3.0.0~\cite{Somenzi2015}, LibBDD~0.5~\cite{Benes2020}, OxiDD~0.6~\cite{Husung2024}, and
Sylvan~1.8.1~\cite{Dijk2016}. Their individual performance relative to Adiar is shown in %
\cref{fig:nested_vs_depth-first}. Out of the 147 instances, 140 are solved by all depth-first BDD
packages, i.e. the remaining 7 instances have at least one BDD package running out of memory (MO) or
time (TO). Adiar solves all of them. Running out of time is most likely due to repeated need for
garbage collection, which essentially is equivalent to an MO. Yet for fairness,
\cref{tab:depth-first} shows the total time for these 140 commonly solved instances. The average
ratio, on the other hand, pertains to all instances solved by the respective BDD package.

The relative running time of BuDDy in \cref{fig:nested_vs_depth-first:buddy} and OxiDD in
\cref{fig:nested_vs_depth-first:oxidd} shows that Adiar's performance can be divided into three
categories: the \emph{small} instances that takes Adiar less than 1~s to solve but is out of its
(current) scope, the \emph{moderate} instances where Adiar needs between 1 and $10^3$~s to solve and
it is up to a constant factor of 4 slower than other BDD packages, and the \emph{large} instances
beyond $10^3$~s where other BDD packages slow down compared to Adiar due to limited internal memory
and repeated garbage collection. While the distinction is not as clear for CUDD in
\cref{fig:nested_vs_depth-first:cudd}, LibBDD in \cref{fig:nested_vs_depth-first:libbdd}, and Sylvan
in \cref{fig:nested_vs_depth-first:sylvan}, they also follow the same trend. This relative
performance is similar to our previous results in
\cite{Soelvsten2022:TACAS,Soelvsten2023:NFM,Soelvsten2023:ATVA,Soelvsten2024:SPIN}. That is, nested
sweeping allows Adiar to compute quantifications at no additional cost to our previous work.

As shown in \cref{tab:depth-first}, Adiar solves the 40 common GoE instances in $2.7$~h and the 100
common QBF instances in $1.25$~h. This makes it as fast as CUDD for both sets of benchmarks and
faster than OxiDD on GoE and Sylvan on QBF. BuDDy is faster for both benchmarks but its use of 32
bits indices limits the maximum BDD size it supports and hence the number of instances it can solve.
As shown in \cref{fig:nested_vs_depth-first:libbdd}, LibBDD is the only BDD package that is
consistently faster than Adiar (ignoring its three MOs). Most likely, this is due to its lack of a
shared unique node table. This improves its cache locality and removes the need for expensive
garbage collections. Hence, LibBDD can either fit its BDD computations into the internal memory (and
it is faster than Adiar) or it aborts.

As is evident from \cref{fig:nested_vs_depth-first:sylvan}, Sylvan is comparatively good at some of
the larger GoE instances, thereby beating all other BDD packages in the total time to solve its
subset of the GoE benchmarks. %
\if\arxiv1%
  As the BDD collapses to $\top$, one may expect this is due to Sylvan skipping the second recursive
  calls to \texttt{exists} if the first recursion resulted in $\top$. Yet, CUDD also includes this
  optimisation without exhibiting the same behaviour. Further investigation is needed to identify
  why Sylvan excels on these instances. Sylvan is also %
\else%
  Furthermore, it is %
\fi%
the only other BDD package able to solve all QBF instances within the given time limit. This is in
parts thanks to its small memory footprint per BDD node \cite{Dijk2016}. Yet, Sylvan requires a
total of $5.0$~h to solve all 102 QBF instances whereas Adiar only needed $3.6$~h, making Sylvan a
factor of 1.4 times slower than Adiar. If Sylvan was given access to more than a single thread, then
it would, of course, be expected to be at least as fast if not much faster than Adiar
\cite{Dijk2016}.

Considering the 140 commonly solved instances are, by definition, not the main focus of the
algorithms in Adiar and that the experiments have been designed in favour of the conventional BDD
packages, it is not surprising that Adiar does not greatly outshine the other BDD packages. Even so,
it is faster than some implementations in some cases -- despite storing BDDs on disk. Most
importantly, it solves more instances than any other BDD package, witnessed by the timeouts in
\cref{fig:nested_vs_depth-first}.

\section{Related Work} \label{sec:related work}

Many other implementations of BDDs also support quantification of multiple variables
\cite{Sanghavi1996,Lind1999,Somenzi2015,Benes2020,Husung2024,Dijk2016}. All these are based on a
nested (inner) operation being accumulated in an (outer) traversal of the input; the nested sweeping
framework achieves the same within the time-forward processing paradigm~\cite{Chiang1995,Arge1995:2}.

\if\arxiv1%
  \subsubsection{CAL:}
\fi%

The CAL~\cite{Sanghavi1996} BDD package (based on \cite{Ochi1993,Ashar1994}) is to the best of our
knowledge the only implementation of BDDs also designed to compute on BDDs whose size exceed main
memory. To do so, it uses breadth-first algorithms that are resolved level by level. For each level
it still follows the conventional approach: a unique node table is used to manage BDD nodes while a
polynomial running time is guaranteed by use of a memoisation table. These per-level hash tables,
both in theory and in practice, put an upper bound on the maximum BDD width that CAL can support
with a certain amount of internal memory \cite{Arge1996}.

\if\arxiv1%
  Its quantification operation also required additional ideas particular to the design of CAL. Since
  it uses a single breadth-first queue for each level, each queue contains requests for both the outer
  and the nested inner traversals. Hence, both can be -- and are -- processed simultaneously
  \cite{Sanghavi1996}. Furthermore, it switches between breadth- and depth-first evaluation of
  subtrees to improve performance: the outer traversal is depth-first for the BDD nodes with to be
  quantified variables and breadth-first otherwise. If the first subtree's quantification makes
  computing the other ones redundant, then all computation of the second is skipped. These depth-first
  steps are also placed in the very same queues as the breadth-first steps; doing so ensures no
  additional random access is introduced.

  By the nature of nested sweeping, our proposed algorithm is, unlike CAL, not easily able to skip
  redundant computations. In \cref{sec:theory:quantify optimisations} we investigate multiple
  promising avenues to achieve similar pruning of redundant computation. %
\else%
  To skip redundant computations, its quantification operation also required additional ideas
  particular to the design of CAL. By the nature of nested sweeping, our proposed algorithm is,
  unlike CAL, not able to skip redundant computations. %
\fi%
Furthermore, the lack of a unique node table in Adiar requires our algorithms to retraverse and copy
the subtrees that are unchanged. Even so, as evident in \cref{sec:experiments}, Adiar with nested
sweeping outperforms CAL by up to several orders of magnitude. Moreover, the I/O-efficient approach
in \cite{Arge1996,Soelvsten2022:TACAS}, and by extension the ones in this work, are, unlike CAL, I/O
efficient despite a BDD's level is wider than main memory.


\if\arxiv1%
  \subsubsection{Distribution Sweeping:}

  In the context of computational geometry, \emph{distribution sweeping}~\cite{Goodrich1993} is an
  I/O-efficient translations of internal memory sweepline algorithms. Here, the recursion is turned
  on its head: the recursive but I/O inefficient data structure is replaced with an I/O-efficient
  list and the iterative algorithm is instead turned into a recursive one. Specifically, all the
  points in the plane are sorted on the $x$-axis and distributed into $M/B$ vertical \emph{strips}
  (see \cref{sec:preliminaries:io} on the I/O-model). After these strips have been solved
  recursively, an $M/B$-way merge procedure both merges and prunes all strips into one while
  simultaneously recreating a vertical sweepline moving across all strips
  \cite{Goodrich1993,Brodal2002}.

  In our case of translating the \texttt{exists} algorithm (see \cref{fig:exists}), we also
  intend to move the recursion out of a data structure, namely out of the BDD. Unlike for
  distribution sweeping, we do not intend to divide-and-conquer the input but instead recurse
  through the dependencies of the algorithm's recursion, e.g.\ between the independent calls to
  \texttt{exists} and the nested \texttt{or} operation that depends on their result. Independent
  recursions are resolved simultaneously with regular time-forward processing sweeps as in
  \cite{Soelvsten2022:TACAS}. Dependencies are handled by moving requests from the priority queue of
  one time-forward processing sweep to the priority queue of another. When all dependencies have
  been moved, the current sweep is paused to then start a \emph{nested sweep} -- the results of
  which are in turn parsed to its dependencies.

\fi%

\section{Conclusions and Future Work} \label{sec:conclusion}

Each sweep in \cite{Soelvsten2022:TACAS} is independent of the others. Using only this approach, one
can only quantify a single variable a time but not multiple at once. In this work, we enable
multi-variable quantification with the \emph{nested sweeping} framework. Here, multiple sweeps work
together: each sweep forwards information within priority queues to itself, its parent, or its child
in a recursion stack.

\if\arxiv0%
  In practice, nested sweeping improves the performance of Adiar.
\else%
  In practice, nested sweeping has improved the total time that Adiar needs to solve our
  quantification benchmarks by $21\%$. On average, it improves the running time of each instance by
  a factor of $1.7$.
\fi%
This allows us to extend our previous results in
\cite{Soelvsten2022:TACAS,Soelvsten2023:NFM,Soelvsten2023:ATVA,Soelvsten2024:SPIN} to Adiar's
quantification operations: ignoring small instances, Adiar is at most 4 times slower than
conventional depth-first implementations. Adiar even outperforms depth-first implementations as they
get closer to the limits of internal memory. As Adiar's nested sweeping algorithms are implemented
on-top of the I/O-efficient data structures that were also used in
\cite{Soelvsten2022:TACAS,Soelvsten2023:NFM,Soelvsten2023:ATVA,Soelvsten2024:SPIN}, its performance
is unaffected by a limited internal memory \cite{Soelvsten2022:TACAS}. For example, whereas
CUDD~\cite{Somenzi2015} could solve 141 out of our 147 benchmark instances in 5.6~h, Adiar needed
only 4.6~h to do the same and it could also solve the remaining 6 instances.
\if\arxiv1%
  On average, Adiar is only $1.3$ times slower than CUDD for the instances that CUDD could solve.%
\fi%
Adiar is also faster, often by several orders of magnitude, than the only other external memory BDD
package, CAL~\cite{Sanghavi1996}.

The nested sweeping framework has already been generalised to pave the way for the implementation of
other multi-recursive BDD operations. We have used this to implement the relational product
\cite{Soelvsten2025:arXiv}; doing so requires several additional optimisations for its variable
relabelling and its combined \texttt{and-exists} \cite{Soelvsten2025:arXiv}. In
\cite{Soelvsten2025:arXiv}, we also provide an evaluation on a large collection of symbolic model
checking experiments. We also hope to use nested sweeping for functional composition and as the
foundation for novel I/O-efficient variable reordering procedures \cite{Soelvsten2025:thesis}.
Finally, nested sweeping opens up the possiblity to create an I/O-efficient implementation of other
types of decision diagrams. For example, both Quantum Multiple-valued Decision
Diagrams~\cite{Miller2006} and Polynomial Boolean Rings~\cite{Brickenstein2009} require nested
sweeps to implement their multiplication operations.

\subsection*{Acknowledgements}

Thanks to the Centre for Scientific Computing, Aarhus, for running our benchmarks on the Grendel
cluster and thanks to Marijn Heule and Randal E. Bryant for suggesting the Garden of Eden
problem and their ideas on how to encode it.

\subsection*{Data Availability Statement}

The source code for all our benchmarks is available on the following GitHub repository at commit
\texttt{51b1375}:
\begin{center}
  \href{https://github.com/ssoelvsten/bdd-benchmark}{github.com/ssoelvsten/bdd-benchmark}
\end{center}
The source code is also available at
\href{https://doi.org/10.5281/zenodo.17552837}{doi:10.5281/zenodo.4718224} on Zenodo. The raw data and
its analysis is available at \href{https://doi.org/10.5281/zenodo.17054026}{doi:10.5281/zenodo.17054026}
.

%
%
%

\DeclareRobustCommand{\VAN}[3]{#2} 

\bibliographystyle{splncs04}
\bibliography{references}

\end{document}